\documentclass[3p]{elsarticle}

\usepackage[utf8]{inputenc}
\usepackage[english]{babel}
\usepackage[babel]{csquotes}

\usepackage{amscd}
\usepackage{amsfonts}
\usepackage{amsmath}
\usepackage{amssymb}
\usepackage[usenames,dvipsnames]{color}
\usepackage{array}
\usepackage{longtable}
\usepackage{xspace}
\usepackage{paralist}
\usepackage{listings}
\usepackage{extarrows}
\usepackage{xifthen}
\usepackage{url}
  \usepackage{mathpartir}
  
\usepackage{amsthm}

\newtheorem{example}{Example}
\newtheorem{lemma}{Lemma}
\newtheorem{remark}{Remark}
\newtheorem{definition}{Definition}

\let\exampleOrig\endexample
\def\endexample{\hspace*{0pt}\hfill$\triangleleft$\exampleOrig}

\usepackage{pgf}
\usepackage{tikz}
\usetikzlibrary{trees,decorations,arrows,automata,shadows,positioning,plotmarks,backgrounds,shapes}
\usetikzlibrary{calc,matrix,fit,petri,decorations.pathmorphing,patterns}

 \pgfdeclarelayer{foreground}
 \pgfdeclarelayer{background}
 \pgfsetlayers{background,main,foreground}

\newcommand{\REFlem}[1]{\text{Lemma~\ref{#1}}}

\newcommand{\REFdef}[1]{Definition~\ref{#1}}

\newcommand{\REFsec}[1]{Section~\ref{#1}}

\newcommand{\deff}{:=}
\newcommand{\BR}[1]{\left( #1 \right)}

\newcommand{\ON}[1]{\operatorname{#1}}

\def\clap#1{\hbox to 0pt{\hss#1\hss}}

\newcommand{\val}[1]{\ensuremath{\mathsf{#1}}}

\newcommand{\DiCases}[4]{\ensuremath{\begin{cases}%
#1&,~#2\\%
#3&,~#4%
\end{cases}}}%
\makeatletter 
\newif\ifFIRST
\newif\ifSECOND
\let\LISTOP\relax
\newcommand{\List}[4][\;]{#3#1%
	\FIRSTtrue
	\@for\i:=#2\do{%
	\ifFIRST\LISTOP{\i}\FIRSTfalse\else,\LISTOP{\i}\fi%
	}%
	#1#4%
	\let\LISTOP\relax
}
\makeatother
\newcounter{DINGLIST}
\stepcounter{DINGLIST}
\newcommand{\markD}[3][\;\;]{\text{\ding{\the\numexpr171+\theDINGLIST}\stepcounter{DINGLIST}}#1#3}

\newcommand{\ZZ}{\textsl{Zu Zeigen:}~\@ifstar\ZZStar\ZZNoStar}
\newcommand{\ZZStar}[1]{\begin{align*}#1\end{align*}}
\newcommand{\ZZNoStar}[1]{\ensuremath{#1}}
\newcommand{\THATIS}{i.e.\xspace}
\newcommand{\SUCHTHAT}{s.t.\xspace}
\newcommand{\IFF}{\unskip~\text{iff}~}
\newcommand{\SHOW}[2][]{Show \ifthenelse{\isempty{#1}}{}{#1, \THATIS, }\ensuremath{#2}:}

\newcommand{\SORRY}[1]{\emph{\color{red}Sorry: #1}\@latex@warning{SORRY: #1}}

\makeatletter

\newcommand{\propNeg}{\@ifstar\propNegStar\propNegNoStar}
\newcommand{\propNegStar}[1]{\ensuremath{\left(\propNegNoStar{#1}\right)}}
\newcommand{\propNegNoStar}[2][\cdot]{\ensuremath{\neg\ifthenelse{\isempty{#2}}{#1}{#2}}}

\newcommand{\propConj}{\@ifstar\propConjStar\propConjNoStar}
\newcommand{\propConjStar}[2]{\ensuremath{\left(\propConjNoStar{#1}{#2}\right)}}
\newcommand{\propConjNoStar}[3][\cdot]{\ensuremath{\ifthenelse{\isempty{#2}}{#1}{#2}\wedge\ifthenelse{\isempty{#3}}{#1}{#3}}}

\newcommand{\propDisj}{\@ifstar\propDisjStar\propDisjNoStar}
\newcommand{\propDisjStar}[2]{\ensuremath{\left(\propDisjNoStar{#1}{#2}\right)}}
\newcommand{\propDisjNoStar}[3][\cdot]{\ensuremath{\ifthenelse{\isempty{#2}}{#1}{#2}\vee\ifthenelse{\isempty{#3}}{#1}{#3}}}

\newcommand{\propImp}{\@ifstar\propImpStar\propImpNoStar}
\newcommand{\propImpStar}[2]{\ensuremath{\left(\propImpNoStar{#1}{#2}\right)}}
\newcommand{\propImpNoStar}[3][\cdot]{\ensuremath{\ifthenelse{\isempty{#2}}{#1}{#2}\Rightarrow\ifthenelse{\isempty{#3}}{#1}{#3}}}

\newcommand{\propAequ}{\@ifstar\propAequStar\propAequNoStar}
\newcommand{\propAequStar}[2]{\ensuremath{\left(\propAequNoStar{#1}{#2}\right)}}
\newcommand{\propAequNoStar}[3][\cdot]{\ensuremath{\ifthenelse{\isempty{#2}}{#1}{#2}\Leftrightarrow\ifthenelse{\isempty{#3}}{#1}{#3}}}

\newcommand{\propXOR}{\@ifstar\propXORStar\propXORNoStar}
\newcommand{\propXORStar}[2]{\ensuremath{\left(\propXORNoStar{#1}{#2}\right)}}
\newcommand{\propXORNoStar}[3][\cdot]{\ensuremath{\ifthenelse{\isempty{#2}}{#1}{#2}\oplus\ifthenelse{\isempty{#3}}{#1}{#3}}}

\newcommand{\AllQ}{\@ifstar\AllQStar\AllQNoStar}
\newcommand{\AllQStar}[3][\;]{\ensuremath{\left(\forall #2#1.#1#3\right)}}
\newcommand{\AllQNoStar}[3][\;]{\ensuremath{\forall #2#1.#1#3}}
\newcommand{\AllQu}{\@ifstar\AllQuStar\AllQuNoStar}
\newcommand{\AllQuStar}[3][\;]{\ensuremath{\left(\forall^{\infty} #2#1.#1#3\right)}}
\newcommand{\AllQuNoStar}[3][\;]{\ensuremath{\forall^{\infty} #2#1.#1#3}}

\newcommand{\ExQ}{\@ifstar\ExQStar\ExQNoStar}
\newcommand{\ExQStar}[3][\;]{\ensuremath{\left(\exists #2#1.#1#3\right)}}
\newcommand{\ExQNoStar}[3][\;]{\ensuremath{\exists #2#1.#1#3}}

\newcommand{\NExQ}{\@ifstar\NExQStar\NExQNoStar}
\newcommand{\NExQStar}[3][\;]{\ensuremath{\left(\nexists #2#1.#1#3\right)}}
\newcommand{\NExQNoStar}[3][\;]{\ensuremath{\nexists #2#1.#1#3}}

\newcommand{\UniqueQ}{\@ifstar\UniqueQStar\UniqueQNoStar}
\newcommand{\UniqueQStar}[3][\;]{\ensuremath{\left(\exists! #2#1.#1#3\right)}}
\newcommand{\UniqueQNoStar}[3][\;]{\ensuremath{\exists! #2#1.#1#3}}

\newcommand{\Set}[2][]{\List[#1]{#2}{\{}{\}}}
\newcommand{\VSet}[2][]{\let\LISTOP\val\List[#1]{#2}{\{}{\}}}

\newcommand{\Tuple}[2][]{\List[#1]{#2}{(}{)}}
\newcommand{\VTuple}[2][]{\let\LISTOP\val\List[#1]{#2}{(}{)}}

\newcommand{\SetComp}[3][]{\{#1#2#1\mid#1#3#1\}}
\newcommand{\SetCompX}[3][]{\left\{#1#2#1\middle\vert#1#3#1\right\}}

\newcommand{\POWERSET}{\@ifstar\POWERSETStar\POWERSETNoStar}
\newcommand{\POWERSETStar}[1]{\ensuremath{\ON{2}^{\ifthenelse{\isempty{#1}}{\cdot}{#1}}}}
\newcommand{\POWERSETNoStar}[1]{\ensuremath{\ON{2}^{\ifthenelse{\isempty{#1}}{\cdot}{#1}}}}

\newcommand{\FINPOWERSET}{\@ifstar\FINPOWERSETStar\FINPOWERSETNoStar}
\newcommand{\FINPOWERSETStar}[1]{\ensuremath{\mathcal{P}_{\ON{fin}}(\ifthenelse{\isempty{#1}}{\cdot}{#1})}}
\newcommand{\FINPOWERSETNoStar}[1]{\ensuremath{\mathcal{P}_{\ON{fin}}\left(\ifthenelse{\isempty{#1}}{\cdot}{#1}\right)}}

\newcommand{\UNION}{\@ifstar\UNIONStar\UNIONNoStar}
\newcommand{\UNIONStar}[2]{\ensuremath{\left(\UNIONNoStar{#1}{#2}\right)}}
\newcommand{\UNIONNoStar}[2]{\ensuremath{\ifthenelse{\isempty{#1}}{\cdot}{#1}\cup\ifthenelse{\isempty{#2}}{\cdot}{#2}}}

\newcommand{\UNIOND}{\@ifstar\UNIONDStar\UNIONDNoStar}
\newcommand{\UNIONDStar}[2]{\ensuremath{\left(\UNIONDNoStar{#1}{#2}\right)}}
\newcommand{\UNIONDNoStar}[2]{\ensuremath{\ifthenelse{\isempty{#1}}{\cdot}{#1}\uplus\ifthenelse{\isempty{#2}}{\cdot}{#2}}}

\newcommand{\SETMINUS}{\@ifstar\SETMINUSStar\SETMINUSNoStar}
\newcommand{\SETMINUSStar}[2]{\ensuremath{\left(\SETMINUSNoStar{#1}{#2}\right)}}
\newcommand{\SETMINUSNoStar}[2]{\ensuremath{\ifthenelse{\isempty{#1}}{\cdot}{#1}\setminus\ifthenelse{\isempty{#2}}{\cdot}{#2}}}

\newcommand{\INTERSECT}{\@ifstar\INTERSECTStar\INTERSECTNoStar}
\newcommand{\INTERSECTStar}[2]{\ensuremath{\left(\INTERSECTNoStar{#1}{#2}\right)}}
\newcommand{\INTERSECTNoStar}[2]{\ensuremath{\ifthenelse{\isempty{#1}}{\cdot}{#1}\cap\ifthenelse{\isempty{#2}}{\cdot}{#2}}}

\newcommand{\CARTPROD}{\@ifstar\CARTPRODStar\CARTPRODNoStar}
\newcommand{\CARTPRODStar}[2]{\ensuremath{\left(\CARTPRODNoStar{#1}{#2}\right)}}
\newcommand{\CARTPRODNoStar}[2]{\ensuremath{\ifthenelse{\isempty{#1}}{\cdot}{#1}\times\ifthenelse{\isempty{#2}}{\cdot}{#2}}}

\newcommand{\FINCOUNT}{\@ifstar\FinCountStar\FinCountNoStar}
\newcommand{\FinCountStar}[1]{\ensuremath{\#(\ifthenelse{\isempty{#1}}{\cdot}{#1})}}
\newcommand{\FinCountNoStar}[1]{\ensuremath{\#\left(\ifthenelse{\isempty{#1}}{\cdot}{#1}\right)}}

\makeatother 

\newcommand{\sconc}{\hspace{-0.6mm}\cdot\hspace{-0.6mm}}

\newcommand{\fun}{\ensuremath{\ON{\rightarrow}}}
\tikzstyle{istate}=[state,initial,initial text=]
\tikzstyle{fistate}=[state,accepting,initial,initial text=]
\tikzstyle{fistateA}=[state,accepting,initial,initial text=,initial where=above]
\tikzstyle{fistateB}=[state,accepting,initial,initial text=,initial where=below]
\tikzstyle{fistateL}=[state,accepting,initial,initial text=,initial where=left]
\tikzstyle{fistateR}=[state,accepting,initial,initial text=,initial where=right]
\tikzstyle{ifstate}=[state,accepting,initial,initial text=]
\tikzstyle{ifstateA}=[state,accepting,initial,initial text=,initial where=above]
\tikzstyle{ifstateB}=[state,accepting,initial,initial text=,initial where=below]
\tikzstyle{ifstateL}=[state,accepting,initial,initial text=,initial where=left]
\tikzstyle{ifstateR}=[state,accepting,initial,initial text=,initial where=right]
\tikzstyle{istateA}=[state,initial,initial text=,initial where=above]
\tikzstyle{istateB}=[state,initial,initial text=,initial where=below]
\tikzstyle{istateL}=[state,initial,initial text=,initial where=left]
\tikzstyle{istateR}=[state,initial,initial text=,initial where=right]
\tikzstyle{fstate}=[state,accepting]
\tikzstyle{SFSautomat}=[->,>=stealth',shorten >=1pt,auto,node distance=2cm,on grid,semithick,inner sep=1pt,bend angle=45]
\newcommand{\SFSAutomatEdge}[5]{\draw[->, thick](#1) edge[#4] node[#5] {\ensuremath{#2}} (#3);}

\newenvironment{propConjA}{\left(\def\unionAtest{1}\begin{array}{@{\if\unionAtest1\gdef\unionAtest{0}\phantom{\wedge}\else\wedge\fi}l@{}}}{\end{array}\right)}

\makeatletter
  \newlength{\SFS@HEIGHT}
  \newlength{\SFS@WIDTH}
  \newcommand{\SplitX}[2]{
	    \settoheight{\SFS@HEIGHT}{$#2$}
	    \settowidth{\SFS@WIDTH}{$#2$}
	    \mbox{\begin{tikzpicture}[baseline=(current bounding box.center)]
	    \node[] (E) at (0,0) {$#1$};
	    \node[inner sep=0pt] (F) at ($(E.south west)+(1ex,-1ex)+(3ex+.5\SFS@WIDTH,-\SFS@HEIGHT)$) {$#2$};
	    \node[] (E) at (0,0) {\phantom{$#1$}};
	    \draw[fill] ($(E.east)+(1ex,0ex)$) circle (.2ex);
	    \draw[-] ($(E.east)+(1ex,0ex)$) -- ($(E.south east)+(1ex,-0.5ex)$) -- ($(E.south west)+(1ex,-0.5ex)$) -- ($(E.south west)+(1ex,-1ex)-(0,\SFS@HEIGHT)$) -- ($(F.west)+(-1ex,0)$);
	    \draw[fill] ($(F.west)+(-1ex,0)$) circle (.2ex);
	    \end{tikzpicture}}}
   \newcommand{\SplitS}[2]{
	    \settoheight{\SFS@HEIGHT}{$#2$}
	    \settowidth{\SFS@WIDTH}{$#2$}
	    \mbox{\begin{tikzpicture}[baseline=(current bounding box.center)]
	    \node[] (E) at (0,0) {$#1$};
	    \node[inner sep=0pt] (F) at ($(E.south west)+(1ex,0.5ex)+(3ex+.5\SFS@WIDTH,-\SFS@HEIGHT)$) {$#2$};
	    \end{tikzpicture}}}

\makeatother

\providecommand{\length}[1]{\lvert#1\rvert}
\providecommand{\lengthw}[1]{\lvert#1\rvert_{\hspace{-0.2mm}_L}}

\newcommand{\trivialN}[1]{\text{trivial}\xspace}

\newcommand{\Nbn}{\ensuremath{\mathbb{N}_{0}}}

\newcommand{\twoup}[1]{\ensuremath{2^{#1}}} 

\newcommand{\Psys}{\ensuremath{P}} 
\newcommand{\Qsys}{\ensuremath{Q}} 
\newcommand{\X}{\ensuremath{X}} 
\newcommand{\Xt}[2]{\ensuremath{\ifthenelse{\isempty{#2}}{X_{\T,#1}}{X_{#2,{\T_{#2}},#1}}}} 
\newcommand{\XT}[1]{\ensuremath{\ifthenelse{\isempty{#1}}{X_\T}{X_{#1,{\T_{#1}}}}}} 
\newcommand{\Xk}[2]{\ensuremath{\ifthenelse{\isempty{#2}}{X_{\TE,#1}}{X_{#2,\TE,#1}}}} 
\newcommand{\XK}[1]{\ensuremath{\ifthenelse{\isempty{#1}}{X_{\TE}}{X_{#1,\TE}}}} 
\newcommand{\Zk}[2]{\ensuremath{\ifthenelse{\isempty{#2}}{Z_{\TE,#1}}{Z_{#2,\TE,#1}}}} 

\newcommand{\x}{\ensuremath{x}} 
 
\newcommand{\Xo}[1]{\ensuremath{X_{#1 0}}}

\newcommand{\Zt}[1]{\ensuremath{\ifthenelse{\isempty{#1}}{Z_t}{Z_{#1,t}}}} 
\newcommand{\ZT}[1]{\ensuremath{\ifthenelse{\isempty{#1}}{Z_T}{Z_{#1,T}}}} 
 
\newcommand{\ZPit}[1]{\ensuremath{\ifthenelse{\isempty{#1}}{Z_t}{Z_{#1,t}}}} 
\newcommand{\ZPiT}[1]{\ensuremath{\ifthenelse{\isempty{#1}}{Z_T}{Z_{#1,T}}}} 
 
\newcommand{\ZTit}[1]{\ensuremath{\ifthenelse{\isempty{#1}}{\breve{Z}_t}{\breve{Z}_{#1,t}}}} 
\newcommand{\ZTiT}[1]{\ensuremath{\ifthenelse{\isempty{#1}}{\breve{Z}_T}{\breve{Z}_{#1,T}}}} 
\newcommand{\z}{\ensuremath{z}}

\newcommand{\w}{\ensuremath{w}}

\newcommand{\tr}{\ensuremath{\delta}}
\newcommand{\trt}{\ensuremath{\psi}}

\newcommand{\T}{\ensuremath{T}} 
\newcommand{\I}{\ensuremath{\mathcal{I}}}

\renewcommand{\ll}[1]{\ensuremath{|_{[#1]}}} 
\newcommand{\lb}[1]{\ensuremath{|_{\langle#1\rangle}}}

\newcommand{\f}[1]{\ensuremath{f\ifthenelse{\isempty{#1}}{}{\Tuple{#1}}}} 
\newcommand{\g}[1]{\ensuremath{g\ifthenelse{\isempty{#1}}{}{\Tuple{#1}}}} 
\newcommand{\h}[1]{\ensuremath{h\ifthenelse{\isempty{#1}}{}{\Tuple{#1}}}} 
\newcommand{\fs}[2]{\ensuremath{f_{#2}\ifthenelse{\isempty{#1}}{}{\Tuple{#1}}}} 
\newcommand{\gs}[2]{\ensuremath{g_{#2}\ifthenelse{\isempty{#1}}{}{\Tuple{#1}}}}  
\newcommand{\hs}[2]{\ensuremath{h_{#2}\ifthenelse{\isempty{#1}}{}{\Tuple{#1}}}}  

\newcommand{\lnc}{\ensuremath{l_t}}

\newcommand{\Beh}{\ensuremath{\mathcal{B}}}

 \newcommand{\Ds}[1]{\ensuremath{\mathcal{D}_{#1}}}

\newcommand{\ESn}[1]{\ensuremath{\ifthenelse{\isempty{#1}}{\Sigma^+_{S}}{\Sigma^+_{S,#1}}}}

\newcommand{\BehS}[1]{\ensuremath{\ifthenelse{\isempty{#1}}{\Beh_{S}}{\Beh_{S,#1}}}}
\newcommand{\BehE}[1]{\ensuremath{\ifthenelse{\isempty{#1}}{\Beh_{E}}{\Beh_{E,#1}}}}

\newcommand{\WT}{\ensuremath{W}}
\newcommand{\W}{\ensuremath{W}}
\newcommand{\D}{\ensuremath{D}}
\newcommand{\V}{\ensuremath{V}}
\newcommand{\statemap}[3]{\ifthenelse{\isempty{#2#3}}{\psi_{#1}}{\psi_{#1}(#2,#3)}}
\newcommand{\statemapPi}[3]{\ifthenelse{\isempty{#2#3}}{\varphi_{#1}}{\varphi_{#1}(#2,#3)}}
\newcommand{\statemapTi}[3]{\ifthenelse{\isempty{#2#3}}{\psi_{#1}}{\psi_{#1}(#2,#3)}}

\newcommand{\CONCAT}[4]{#1\wedge^{#2}_{#3}#4}

\newcommand{\Xx}[2]{\ensuremath{\chi_{#1}\ifthenelse{\isempty{#2}}{}{(#2)}}}
\newcommand{\Xxp}[2]{\ensuremath{\overline{\chi}_{#1}\ifthenelse{\isempty{#2}}{}{(#2)}}}
\newcommand{\Xxr}[3]{\ensuremath{\chi_{#1}^{#2}\ifthenelse{\isempty{#3}}{}{(#3)}}}
\newcommand{\Xxrp}[3]{\ensuremath{\overline{\chi}_{#1}^{#2}\ifthenelse{\isempty{#3}}{}{(#3)}}}

\newcommand{\signalmap}{\phi}
\newcommand{\TE}{\ensuremath{{T_E}}}

\newcommand{\E}{\ensuremath{\Sigma}}
\newcommand{\Ep}{\ensuremath{\Sigma^{\signalmap}}}
\newcommand{\ES}[1]{\ensuremath{\ifthenelse{\isempty{#1}}{\Sigma_{S}}{\Sigma_{S,#1}}}}
\newcommand{\EpS}[1]{\ensuremath{\ifthenelse{\isempty{#1}}{\Ep_{S}}{\Sigma^{\signalmap_{#1}}_{S,#1}}}}
\newcommand{\EE}[1]{\ensuremath{\ifthenelse{\isempty{#1}}{\Sigma_{E}}{\Sigma_{E,#1}}}}
\newcommand{\ESm}[1]{\ensuremath{\ifthenelse{\isempty{#1}}{\Sigma_{\psi}}{\Sigma_{\psi,#1}}}}

\newcommand{\El}{\ensuremath{\Sigma^{l}}}

\newcommand{\ElMax}{\ensuremath{\Sigma^{l^\Uparrow}}}
\newcommand{\ElaMax}{\ensuremath{\Sigma^{l^\uparrow}}}

\newcommand{\EplMaxS}[1]{\ensuremath{\ifthenelse{\isempty{#1}}{\Sigma^{\phi,l^\uparrow}_{S}}{\Sigma^{\phi_{#1},l^\uparrow}_{S,#1}}}}
\newcommand{\EplMaxSp}[1]{\ensuremath{\ifthenelse{\isempty{#1}}{\overline{\Sigma}^{\phi,l^\uparrow}_{S}}{\overline{\Sigma}^{\phi_{#1},l^\uparrow}_{S,#1}}}}
\newcommand{\EplncMaxS}[1]{\ensuremath{\ifthenelse{\isempty{#1}}{\Sigma^{\phi,\lnc^\uparrow}_{S}}{\Sigma^{\phi_{#1},\lnc^\uparrow}_{S,#1}}}}
\newcommand{\EplsMaxS}[2]{\ensuremath{\ifthenelse{\isempty{#1}}{\Sigma^{\phi,{#2}^\uparrow}_{S}}{\Sigma^{\phi_{#1},{#2}^\uparrow}_{S,#1}}}}

\newcommand{\ElMaxS}[1]{\ensuremath{\ifthenelse{\isempty{#1}}{\Sigma^{l^\Uparrow}_{S}}{\Sigma^{l^\Uparrow}_{S,#1}}}}
\newcommand{\ElaMaxS}[1]{\ensuremath{\ifthenelse{\isempty{#1}}{\Sigma^{l^\uparrow}_{S}}{\Sigma^{l^\uparrow}_{S,#1}}}}
\newcommand{\ElMaxSp}[1]{\ensuremath{\ifthenelse{\isempty{#1}}{\overline{\Sigma}^{l^\uparrow}_{S}}{\overline{\Sigma}^{l^\uparrow}_{S,#1}}}}
\newcommand{\ElncMaxS}[1]{\ensuremath{\ifthenelse{\isempty{#1}}{\Sigma^{\lnc^\uparrow}_{S}}{\Sigma^{\lnc^\uparrow}_{S,#1}}}}
\newcommand{\ElsMaxS}[2]{\ensuremath{\ifthenelse{\isempty{#1}}{\Sigma^{{#2}^\uparrow}_{S}}{\Sigma^{{#2}^\uparrow}_{S,#1}}}}

\newcommand{\ElE}[1]{\ensuremath{\ifthenelse{\isempty{#1}}{\Sigma^{l}_{E}}{\Sigma^{l}_{E,#1}}}}
\newcommand{\ElEp}[1]{\ensuremath{\ifthenelse{\isempty{#1}}{\Sigma^{l}_{E}}{\overline{\Sigma}^{l}_{E,#1}}}}
\newcommand{\ElncE}[1]{\ensuremath{\ifthenelse{\isempty{#1}}{\Sigma^{\lnc}_{E}}{\Sigma^{\lnc}_{E,#1}}}}
\newcommand{\ElMaxE}[1]{\ensuremath{\ifthenelse{\isempty{#1}}{\Sigma^{l^\uparrow}_{E}}{\Sigma^{l^\uparrow}_{E,#1}}}}
\newcommand{\ElMaxEp}[1]{\ensuremath{\ifthenelse{\isempty{#1}}{\overline{\Sigma}^{l^\uparrow}_{E}}{\overline{\Sigma}^{l^\uparrow}_{E,#1}}}}
\newcommand{\ElncMaxE}[1]{\ensuremath{\ifthenelse{\isempty{#1}}{\Sigma^{\lnc^\uparrow}_{E}}{\Sigma^{\lnc^\uparrow}_{E,#1}}}}
\newcommand{\ElsMaxE}[2]{\ensuremath{\ifthenelse{\isempty{#1}}{\Sigma^{{#2}^\uparrow}_{E}}{\Sigma^{{#2}^\uparrow}_{E,#1}}}}

\newcommand{\Behl}{\ensuremath{\Beh^{l}}}

\newcommand{\BehlMax}{\ensuremath{\Beh^{l^\Uparrow}}}
\newcommand{\BehlaMax}{\ensuremath{\Beh^{l^\uparrow}}}

\newcommand{\BehlMaxS}[1]{\ensuremath{\ifthenelse{\isempty{#1}}{\Beh^{l^\Uparrow}_{S}}{\Beh^{l^\Uparrow}_{S,#1}}}}
\newcommand{\BehlaMaxS}[1]{\ensuremath{\ifthenelse{\isempty{#1}}{\Beh^{l^\uparrow}_{S}}{\Beh^{l^\uparrow}_{S,#1}}}}
\newcommand{\BehlMaxSp}[1]{\ensuremath{\ifthenelse{\isempty{#1}}{\overline{\Beh}^{l^\uparrow}_{S}}{\overline{\Beh}^{l^\uparrow}_{S,#1}}}}
\newcommand{\BehlncMaxS}[1]{\ensuremath{\ifthenelse{\isempty{#1}}{\Beh^{\lnc^\uparrow}_{S}}{\Beh^{\lnc^\uparrow}_{S,#1}}}}
\newcommand{\BehlsMaxS}[2]{\ensuremath{\ifthenelse{\isempty{#1}}{\Beh^{{#2}^\uparrow}_{S}}{\Beh^{{#2}^\uparrow}_{S,#1}}}}

\newcommand{\BehlE}[1]{\ensuremath{\ifthenelse{\isempty{#1}}{\Beh^{l}_{E}}{\Beh^{l}_{E,#1}}}}
\newcommand{\BehlncE}[1]{\ensuremath{\ifthenelse{\isempty{#1}}{\Beh^{\lnc}_{E}}{\Beh^{\lnc}_{E,#1}}}}
\newcommand{\BehlMaxE}[1]{\ensuremath{\ifthenelse{\isempty{#1}}{\Beh^{l^\uparrow}_{E}}{\Beh^{l^\uparrow}_{E,#1}}}}
\newcommand{\BehlMaxEp}[1]{\ensuremath{\ifthenelse{\isempty{#1}}{\overline{\Beh}^{l^\uparrow}_{E}}{\overline{\Beh}^{l^\uparrow}_{E,#1}}}}
\newcommand{\BehlncMaxE}[1]{\ensuremath{\ifthenelse{\isempty{#1}}{\Beh^{\lnc^\uparrow}_{E}}{\Beh^{\lnc^\uparrow}_{E,#1}}}}
\newcommand{\BehlsMaxE}[2]{\ensuremath{\ifthenelse{\isempty{#1}}{\Beh^{{#2}^\uparrow}_{E}}{\Beh^{{#2}^\uparrow}_{E,#1}}}}

\newcommand{\BehVlMaxS}[1]{\ensuremath{\ifthenelse{\isempty{#1}}{\projState{\V}{\Beh^{l^\uparrow}_{S}}}{\projState{\V}{\Beh^{l^\uparrow}_{S,#1}}}}}

\newcommand{\BehDlMaxS}[1]{\ensuremath{\ifthenelse{\isempty{#1}}{\projState{\D}{\Beh^{l^\uparrow}_{S}}}{\projState{\D}{\Beh^{l^\uparrow}_{S,#1}}}}}

\newcommand{\EoMaxS}[1]{\ensuremath{\ifthenelse{\isempty{#1}}{\Sigma^{1^\uparrow}_{S}}{\Sigma^{1^\uparrow}_{S,#1}}}}

\newcommand{\projState}[2]{\pi_{#1}(#2)}

\begin{document}

\begin{frontmatter}

\title{Asynchronous $l$-Complete Approximations}

\author[csg,fn1]{Anne-Kathrin Schmuck}
\ead{a.schmuck@control.tu-berlin.de}
\author[csg,mpi]{J\"org Raisch}
\ead{raisch@control.tu-berlin.de}
\fntext[fn1]{Corresponding author, phone: 0049-(0)30-314-24094}
\address[csg]{Control Systems Group, Technische Universität Berlin, Einsteinufer 17, 10587 Berlin, Germany}
\address[mpi]{Max Planck Institute
for Dynamics of Complex Technical Systems, Sandtorstraße 1,
39106 Magdeburg, Germany}

 \begin{abstract}
This paper extends the $l$-complete approximation method developed for time invariant systems to a larger system class, ensuring that the resulting approximation can be realized by a finite state machine. 
To derive the new abstraction method, called asynchronous $l$-complete approximation, an asynchronous version of the well-known concepts of state property, memory span and $l$-completeness is introduced, extending the behavioral systems theory in a consistent way.
\end{abstract}

\begin{keyword}
finite state abstraction, $l$-complete approximations, state property, behavioral systems theory
\end{keyword}
\end{frontmatter}

\section{Introduction}

Real life control problems for large scale systems are very challenging due to numerous interactions between different components and usually tight performance requirements. One way to reduce the complexity of those control problems is to introduce different control layers using a well defined abstraction of the plant. Usually, the top control layer will enforce high level specifications, such as interconnection or safety requirements, typically expressible by regular languages.
With this specification type supervisory control theory (SCT) \cite{RamWon,RamWon1989} can be used to synthesize a correct by design control system if the abstracted plant model can be represented by a regular language as well.\\
Using this well known result, many abstraction techniques, e.g. \cite{MoorRaisch1999,AlurHenzingerLaffarrierePappas2000,MoorRaischYoung2002, Pappas2003,TabuadaPappas2003,Tabuada2008,TabuadaBook}, have been developed to generate a regular language representation of the plant model. The approach by Moor and Raisch \cite{MoorRaisch1999}, called $l$-complete approximation, is distinct in two ways:
\begin{inparaenum}[(i)]
 \item the accuracy of the abstracted system can be adjusted during construction without adjusting the external signal space and
 \item the behavioral framework \cite{Willems1989,Willems1991} is used to model the plant, allowing for infinite signals with eventuality properties.
\end{inparaenum}
If the external signal space is finite and these signals evolve along the non-negative, discrete time axis $\Nbn$, the plant behavior is a so called $\omega$-language. Even though SCT cannot be directly applied to $\omega$-languages, it was shown in \cite{MoorRaisch1999}, and recently generalized in \cite{MoorSchmidtWittmannIFAC2011}, that for $\omega$-languages realizable by finite state machines 
(FSM),
a variant of SCT can be used to synthesize a minimally restrictive controller for specifications representable by the closure of a regular language (for details, see \cite{MoorSchmidtWittmannIFAC2011} and the references therein).\\
In \cite{MoorRaisch1999} and subsequent papers, $l$-complete approximations were only defined for time invariant systems, i.e., systems that are invariant w.r.t.\ the backward time shift of signals. As a slight extension, \cite[p.51]{MoorPhd} also considers 
systems which are time invariant after a finite start-up phase.\\
As pointed out in \cite[p.44]{MoorPhd}, for systems with time axis $\Nbn$, the $l$-completeness property for time invariant systems used in \cite{MoorRaisch1999,MoorPhd} is slightly weaker than the original definition by J.C.Willems \cite{Willems1989}. 
This implies that the strongest $l$-complete approximation suggested in \cite{MoorRaisch1999} is also $l$-complete in the sense of \cite{Willems1989}, but not necessarily the \emph{strongest} $l$-complete approximation in the sense of \cite{Willems1989}.\\
To resolve this inconsistency and to consider a larger system class, we extend the construction of strongest $l$-com\-plete approximations to not necessarily time invariant systems, 
and ensure that these approximations can still be realized by FSMs.\\ 
As a first step, in Sections~\ref{sec:lcom} and \ref{sec:statespace}, we introduce a straightforward extension of the existing approximation method to not necessarily time invariant dynamical systems, ensuring $l$-com\-pleteness in the sense of \cite{Willems1989}. We show in \REFsec{sec:FSMRepresentations} that the constructed abstractions do generally not allow for an FSM realization since they require a time dependent next state relation. 
Intuitively, a system is realizable by an FSM if it allows for concatenation of state trajectories that reach  the same state asynchronously (i.e., at different times), as used in the context of state maps by Julius and van der Schaft \cite{JuliusSchaft2005,JuliusPhdThesis2005}. 
To emphasize that this property does not imply and is not implied by the time invariance property of behavioral systems, we call it \textit{asynchronous state property} and formalize it in \REFsec{sec:AsyncProp}.
Then we can introduce an \textit{asynchronous $l$-completeness} property, since the state and the $l$-completeness property are strongly related.
This leads to a new approximation technique introduced in \REFsec{sec:alcomplapprox}, which is referred to as \textit{asynchronous $l$-complete approximation} and which ensures that the resulting abstraction can be realized by an FSM.

\section{Preliminaries}\label{sec:prelim}

In the behavioral framework (e.g., \cite{Willems1991}) a \textit{dynamical system} is given by $\E=\Tuple{\T,\WT,\Beh}$, consisting of the time axis $\T$ (in this paper: $\T=\Nbn$), the signal space $\WT$ and the behavior of the system $\Beh\subseteq\WT^{\T}$, where 
$\WT^{\T}\deff\SetComp{w}{w:\T\fun\WT}$ is the set of all  
\textit{signals}
taking values in $\WT$. 
 Let $\I$ be a bounded interval on $\Nbn$, then $\WT^\I\deff\SetComp{w}{w:\I\fun\WT}$ is the set of \textit{signals on $\I$} taking values in $\WT$. Given time instants $t_1,t_2\in\Nbn,~t_1\leq t_2$, we say that the string $w\in\W^{[t_1,t_2]}$ is of length $\lengthw{\w}=t_2-t_1+1$. 
 Furthermore, $\w|_{\I}$ is the \textit{restriction} of the map $w:\Nbn\fun\WT$ to the domain $\I$.
 $\Beh|_{\I}\subseteq\WT^\I$ denotes the restriction of all signals in $\Beh$ to $\I$ and we define\footnote{Throughout this paper we use the notation "$\AllQ{\cdot}{\cdot}$", meaning that all statements after the dot hold for all variables in front of the dot. "$\ExQ{\cdot}{\cdot}$" is interpreted analogously. }
 $\AllQ{t_1,t_2\in\Nbn,~t_1< t_2}{\w\ll{t_2,t_1}=\lambda}$, where $\lambda$ denotes the \textit{empty string} with $\lengthw{\lambda}=0$. 
Now let $\WT=\WT_1\times\WT_2$ be a product space. Then the \textit{projection} of a signal $w\in\WT^{\Nbn}$ to $\WT_1$ is given by $\projState{\WT_1}{w}\deff\SetComp{w_1\in\WT_1^{\Nbn}}{\ExQ{w_2\in\WT_2^{\Nbn}}{w=\Tuple{w_1,w_2}}}$ and $\projState{\WT_1}{\Beh}$ denotes the projection of all signals in the behavior to $\WT_1$. 
Given two signals $w_1,w_2\in\WT^{\Nbn}$ and two time instants $t_1,t_2\in\Nbn$, the \textit{concatenation}  $w_3=\CONCAT{w_1}{t_1}{t_2}{w_2}$ is given by
\begin{equation*}
\AllQ{t\in\Nbn}{w_3(t)=\DiCases{w_1(t)}{t< t_1}{w_2(t-t_1+t_2)}{t\geq t_1}},
\end{equation*}
where we denote $\CONCAT{\cdot}{t}{t}{\cdot}$ by $\CONCAT{\cdot}{}{t}{\cdot}$. Furthermore, the concatenation of their restrictions ${w_1'=w_1\ll{0,t_1}}$ and  ${w_2'=w_2\ll{0,t_2}}$ is defined as 
$\w_1'\sconc\w_2':=\BR{\CONCAT{w_1}{t_1+1}{0}{w_2}}\ll{0,t_1+t_2+1}$. This corresponds to the standard concatenation of finite strings. Furthermore, for a finite string $w=\nu_0\nu_1\hdots\nu_l$ we denote the restriction of $w$ by $w\lb{i,j}:=\nu_i\hdots\nu_j$ with ${0\leq i\leq j\leq l}$.
Following \citep[Def. II.3]{Willems1991}, we define the \textit{backward shift operator} $\sigma^t$  \SUCHTHAT $\AllQ{t,k\in\Nbn}{(\sigma^t f)(k)\deff(f(k+t))}$
and say that $\E$ is \textit{time invariant} if
$\sigma\Beh\subseteq \Beh$.
We call $\E$ \textit{strictly time invariant} if $\sigma \Beh=\Beh$.

\section{$l$-completeness and $l$-complete approximation}\label{sec:lcom}
When reasoning about systems with infinite time axis one has to distinguish between local and eventuality properties. Local properties can be evaluated on a finite time interval whereas eventuality properties can only be evaluated after infinite time. Systems whose behavior can be fully described by local properties are called complete \citep[Def. II.4]{Willems1991}; formally, $\E$ is said to be \textbf{complete} if
\begin{equation}\label{equ:complete_old}
\propAequ{\AllQ*{t_1,t_2\in\Nbn,t_1\leq t_2}{w\ll{t_1,t_2}\in \Beh\ll{t_1,t_2}}}{w\in\Beh}.
\end{equation}
It is easy to show that \eqref{equ:complete_old} is equivalent to 
\begin{equation}
\propAequ{\AllQ*{\tau\in\Nbn}{w\ll{0,\tau}\in \Beh\ll{0,\tau}}}{w\in\Beh},
\end{equation}
which is also known as $\omega$-closedness \cite{MoorSchmidtWittmannIFAC2011}.\\
In the special case where the behavior can be fully described by local properties evaluated on time intervals of length $l+1$, $l\in\Nbn$, the system is called \textbf{$l$-complete} \citep[p.184]{Willems1989}, formally 
\begin{equation}\label{equ:def:lcompl}
\propAequ{\AllQ*{t\in\Nbn}{w\ll{t,t+l}\in \Beh\ll{t,t+l}}}{w\in\Beh}.
\end{equation}
To generate some intuition for the $l$-completeness property, 
we 
define $\Ds{l+1}=\bigcup_{t'\in\Nbn}\Beh\ll{t',t'+l}$ to be the set of all finite strings representing the restriction of admissible signals to a time interval of length $l+1$. Now 
consider the following gedankenexperiment: assume playing a sophisticated domino game where $\Ds{l+1}$ is the set of dominos. Pick the first domino from the set $\Beh\ll{0,l}$ and append one domino from the set $\Beh\ll{1,1+l}$ if the last $l$ symbols of the first domino are equivalent to the first $l$ symbols of the second domino. Playing the domino game arbitrarily long and with all possible initial conditions and domino combinations, we get the set $\Behl{}$ containing all signals that satisfy the left side of \eqref{equ:def:lcompl}. If the system is $l$-complete
we have $\Beh=\Behl{}$,  emphasizing that all valid signals can be fully described by a local property.

\begin{example}\label{exp:lcomplete}
\normalfont
 Consider the system
 \begin{align}
  &\E=\Tuple{\Nbn,\WT,\Beh} \quad\SUCHTHAT\label{equ:exp:lcomplete}\\
  &\Beh=\Set{aaab(aab)^\omega, aab(aab)^\omega, ab(aab)^\omega, b(aab)^\omega},\notag
 \end{align}
where $(\cdot)^\omega$ denotes the infinite repetition of the respective string. Observe that $\E$ is time invariant, but not strictly time invariant, since 
\begin{equation*}
aaab(aab)^\omega\notin\sigma\Beh=\Set{aab(aab)^\omega, ab(aab)^\omega, b(aab)^\omega, (aab)^\omega}. 
\end{equation*}
 Using $l=1$ we get the domino set
\begin{equation}\label{equ:exp:lcomplete:1}
 \AllQ{t\in\Nbn}{\Beh\ll{t,t+1}=\Beh\ll{0,1}=\Set{aa,ab,ba}}.
\end{equation}
As depicted in Figure~\ref{fig:domino_l1}, we can start the domino game with the piece $ba$ and append a piece that starts with an $a$, e.g., $aa$. Observe that the signal constructed in Figure~\ref{fig:domino_l1}, i.e., $w=baaab...$, is not allowed in \eqref{equ:exp:lcomplete} since not more than two sequential $a$'s can occur for $t>0$. However, we can of cause construct all signals $w\in\Beh$ using the outlined domino game. This implies that 
 \begin{inparaenum}[(i)]
  \item the system $\E$ in \eqref{equ:exp:lcomplete} is not $1$-complete and
  \item the domino game constructs a behavior $\Beh^1$ that is larger than the one in \eqref{equ:exp:lcomplete}, i.e., $\Beh^1\supset\Beh$.
  \end{inparaenum}
Now, increasing $l$ to $l=2$ gives the following set of domino pieces
 \begin{align}
  &\Beh\ll{0,2}=\Set{aaa,aab,aba,baa},\label{equ:exp:lcomplete:2}\\
  &\AllQ{t>0}{\Beh\ll{t,t+2}}=\Beh\ll{1,3}=\Set{aab,aba,baa}.\notag
 \end{align}
Playing the domino game with these sets results, for example, in the signal depicted in Figure~\ref{fig:domino_l2}, where always two symbols are required to match. Observe that after the first piece we are only allowed to pick from the set $\Beh\ll{1,3}$. This prevents the occurrence of more than two sequential $a$'s since the domino $aaa$ cannot be attached. We get $\Beh^2=\Beh$, i.e., the system $\E$ in \eqref{equ:exp:lcomplete} is $2$-complete.
\end{example}
\begin{figure}[htb]
\begin{minipage}{0.45\linewidth}
\begin{center}
 \begin{tikzpicture}
 \begin{pgfonlayer}{background}
   \path      (0,0) node (o) {
      \includegraphics[width=0.47\linewidth]{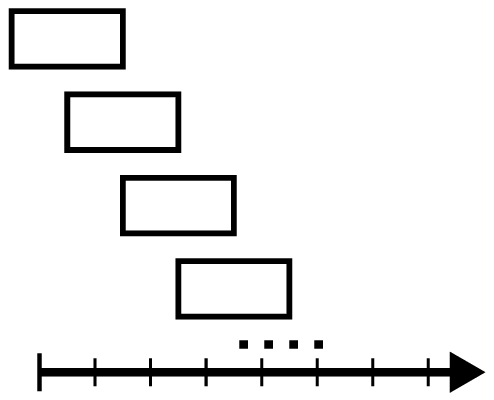}};
 \end{pgfonlayer}
 \begin{pgfonlayer}{foreground}
  \path (o.north west)+(0.5,-0.35) node (i) {$~b~\,a$};
  \path (i.north west)+(0.85,-0.85) node (ii){$~a~\,a$};
  \path (ii.north west)+(0.85,-0.8) node (iii){$~a~\,a$};
  \path (iii.north west)+(0.85,-0.8) node (iiii){$~a~\,b$};
  \path (o.south west)+(0.15,-0.05) node (t0) {$0$};
  \path (t0)+(0.9,0) node (t2) {$2$};
  \path (t2)+(0.85,0) node (t4) {$4$};
  \path (t4)+(0.9,0) node (t6) {$6$};
  \path (t6)+(0.9,0) node (t6) {$\Nbn$};
 \end{pgfonlayer}
 \end{tikzpicture}
 \end{center}
  \vspace{-0.4cm}
 \caption{Domino game for ${l=1}$ in Example~\ref{exp:lcomplete}.}\label{fig:domino_l1}
 \end{minipage}\hfill
 \begin{minipage}{0.45\linewidth}
\begin{center}
 \begin{tikzpicture}
 \begin{pgfonlayer}{background}
   \path      (0,0) node (o) {
      \includegraphics[width=0.47\linewidth]{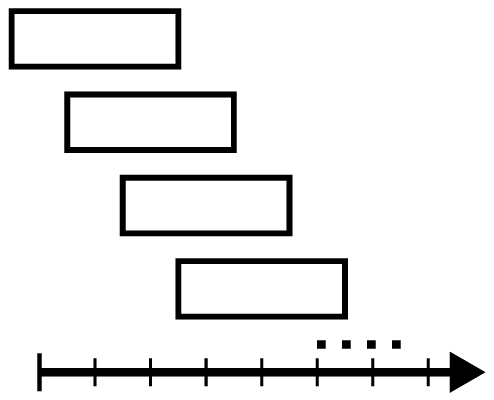}};
 \end{pgfonlayer}
\begin{pgfonlayer}{foreground}
  \path (o.north west)+(0.7,-0.35) node (i) {$~b~\,a~\,a$};
  \path (i.north west)+(1.0,-0.85) node (ii){$~a~\,a~\,b$};
  \path (ii.north west)+(1.0,-0.85) node (iii){$~a~\,b~\,a$};
  \path (iii.north west)+(1.1,-0.85) node (iiii){$b\,~a~\,a$};
   \path (o.south west)+(0.15,-0.05) node (t0) {$0$};
  \path (t0)+(0.9,0) node (t2) {$2$};
  \path (t2)+(0.85,0) node (t4) {$4$};
  \path (t4)+(0.9,0) node (t6) {$6$};
  \path (t6)+(0.9,0) node (t6) {$\Nbn$};
 \end{pgfonlayer}
 \end{tikzpicture}
 \end{center}
 \vspace{-0.4cm}
 \caption{Domino game for ${l=2}$ in Example~\ref{exp:lcomplete}.}\label{fig:domino_l2}
 \end{minipage}
 \end{figure}

As a special case it can be shown that the behavior of an $l$-complete system $\E$ can by fully described by the initial signal pieces $\Beh\ll{0,l}$ if $\E$ is strictly time invariant.

\begin{lemma}\label{lem:lcnessTI}
Let $\E=\Tuple{\Nbn,\WT,\Beh}$ be a strictly time invariant dynamical system and $l\in\Nbn$. Then $\E$ is $l$-complete \IFF 
\begin{equation}\label{equ:lem:lcnessTI:neu}
\propAequ{\AllQ*{t\in\Nbn}{w\ll{t,t+l}\in \Beh\ll{0,l}}}{w\in\Beh}.
\end{equation}
\end{lemma}

\begin{proof}
If $\E$ is strictly time invariant, i.e., $\sigma\Beh=\Beh$, then $\AllQ{t\in\Nbn}{\sigma^t \Beh=\Beh}$, hence
$\AllQ{t\in\Nbn}{\Beh\ll{t,t+l}=\Beh\ll{0,l}}$. Therefore, \eqref{equ:lem:lcnessTI:neu} and \eqref{equ:def:lcompl} are equivalent.
\end{proof}

\begin{remark}\label{rem:AnmerkungMoorRaisch}
For time invariant systems that are not strictly time invariant observe that
$\ExQ{t\in\Nbn}{\Beh\ll{t,t+l}\subset \Beh\ll{0,l}}$, implying that in this case \eqref{equ:lem:lcnessTI:neu} and \eqref{equ:def:lcompl} are not equivalent.
More specifically, the class of time invariant systems satisfying \eqref{equ:lem:lcnessTI:neu} is larger than the class of systems satisfying \eqref{equ:def:lcompl}.
Therefore, the definition of $l$-completeness via \eqref{equ:lem:lcnessTI:neu}, as used in \citep[Def.8]{MoorRaisch1999} and subsequent papers, does formally only coincide with the original definition by J.C.Willems \citep[Sec.1.4.1]{Willems1989} for strictly time invariant systems.
\end{remark}

The set $\Behl{}$ generated in the outlined domino game also matches the behavior of the system $\E$, if $\E$ is $r$-complete with $r\leq l$, since using larger dominos cannot lead to a richer behavior. Furthermore, as already shown in Example~\ref{exp:lcomplete}, we will always get $\Behl{}\supseteq\Beh$ even if the system is not complete at all, since using less information in the domino game generates more freedom in constructing signals.
Formalizing this idea, following
\citep[Def.9]{MoorRaisch1999} we say that $\El{}=\Tuple{\Nbn,\W,\Behl{}}$ is an \textbf{$l$-complete approximation} of $\E=\Tuple{\Nbn,\WT,\Beh}$, if
\begin{inparaenum}[(i)]
 \item $\El{}$ is $l$-complete and 
 \item $\Behl{}\supseteq\Beh{}$.
\end{inparaenum}
Furthermore, $\ElMax{}=\Tuple{\Nbn,\W,\BehlMax{}}$ is the \textbf{strongest $l$-complete approximation} of $\E{}
$, if 
\begin{inparaenum}[(i)]
 \item $\ElMax{}$ is an $l$-complete approximation of $\E{}$ and
 \item for any $l$-complete approximation $\E{}'=\Tuple{\Nbn,\W,\Beh{}'}$ of $\E{}$ it holds that $\BehlMax{}\subseteq\Beh{}'$.
\end{inparaenum}

\begin{remark}
 Note that (strongest) $l$-complete approximations as defined above only coincide with (strongest) $l$-complete approximations introduced in \cite{MoorRaisch1999}, if the underlying system is strongly time invariant. This is an immediate consequence of Remark~\ref{rem:AnmerkungMoorRaisch}.
\end{remark}

Generalizing the results in \citep[Prop.10]{MoorRaisch1999} to the $l$-completeness definition in \eqref{equ:def:lcompl} shows that the behavior $\Behl{}$ constructed in the outlined domino game is the behavior of the strongest $l$-complete approximation, $\BehlMax{}$.

\begin{lemma}\label{lem:constructElMax_general}
 Let $\E{}=\Tuple{\Nbn,\W,\Beh}$ be a dynamical system. Then the unique strongest $l$-complete approximation of $\E{}$ is given by  $\ElMax{}=\Tuple{\Nbn,\W,\BehlMax{}}$, with
 \begin{equation}\label{equ:lem:constructElMax_general:2}
 \BehlMax{}:=\SetCompX{\w\in\W^{\Nbn}}{\AllQ{t\in \Nbn}{\w\ll{t,t+l}\in \Beh\ll{t,t+l}}}.
 \end{equation}
 Furthermore, if $\E$ is strictly time invariant then 
  \begin{equation}\label{equ:lem:constructElMax_general:3}
 \BehlMax{}=\SetCompX{\w\in\W^{\Nbn}}{\AllQ{t\in \Nbn}{\w\ll{t,t+l}\in \Beh\ll{0,l}}}.
 \end{equation}
\end{lemma}

\begin{proof}
\begin{inparaenum}[(i)]
 \item  $\ElMax{}$ is $l$-complete as \eqref{equ:lem:constructElMax_general:2} implies $\BehlMax{}\ll{t,t+l}=\Beh\ll{t,t+l}$, hence
  $\propAequ{w\in\BehlMax{}}{\AllQ*{t\in\Nbn}{w\ll{t,t+l}\in \BehlMax{}\ll{t,t+l}}}$.\\
 \item ${\Beh\subseteq\BehlMax{}}$ holds, as ${w\in\Beh{}}$ implies $\AllQ{t\in\Nbn}{w\ll{t,t+l}\in \Beh\ll{t,t+l}}$, hence ${w\in\BehlMax{}}$ from \eqref{equ:lem:constructElMax_general:2}.\\
 \item For any $l$-complete approximation $\E{}'=\Tuple{\Nbn,\W,\Beh{}'}$ of $\E$ the inclusion $\Beh{}\subseteq\Beh{}'$ and therefore $\Beh{}\ll{t,t+l}\subseteq\Beh{}'\ll{t,t+l}$ holds. 
 Hence, using \eqref{equ:lem:constructElMax_general:2}, $w\in\BehlMax{}$ implies $\AllQ{t\in\Nbn}{w\ll{t,t+l}\in \Beh'\ll{t,t+l}}$ and therefore $w\in\Beh'{}$ since $\E{}'$ is $l$-complete. \\
\end{inparaenum}
Now (i)-(iii) imply that $\ElMax{}$ is a strongest $l$-complete approximation.
Finally, $\ElMax{}$ is unique as (iii) implies that $\BehlMax{}$ is the unique smallest element of the set $\Set{\Beh'}$ containing the behaviors of all $l$-complete approximations $\E{}'=\Tuple{\Nbn,\W,\Beh{}'}$ of $\E$.\\
The second part of the lemma follows directly from \eqref{equ:lem:lcnessTI:neu} in \REFlem{lem:lcnessTI}.
\end{proof}

\begin{example}\label{exp:Beh1strictlyTI}
\normalfont
 As a consequence of $\REFlem{lem:constructElMax_general}$, the behaviors $\Beh^1$ and $\Beh^2$ constructed in Example~\ref{exp:lcomplete} characterize the strongest $1$-complete and the strongest $2$-complete approximation of the system in \eqref{equ:exp:lcomplete}, respectively. 
\end{example}

\section{State Space Systems}\label{sec:statespace}
To represent a behavior, internal variables can be useful. 
Following \citep[Def.1.2]{Willems1989}, a \textbf{dynamical system with internal signal space} $X$ is defined by $\ES{}=\Tuple{\Nbn,\WT,\X,\BehS{}}$ with $\BehS{}\subseteq(\W\times\X)^{\Nbn}$. The internal variables are called states, if the axiom of state holds, i.e., all relevant information from the past and present necessary to decide on the possible future evolution of the system is captured by the current value of the internal variable. Formally, a system $\ES{}=\Tuple{\Nbn,\WT,\X,\BehS{}}$ is a \textbf{state space dynamical system} \citep[p.185]{Willems1989}, if
\begin{equation}\label{equ:StateSpaceDynamicalSystem:1}
 \AllQ{\Tuple{w_1,x_1},\Tuple{w_2,x_2}\in\BehS{}, t\in\Nbn}{\propImp*{x_1(t)=x_2(t)}{\CONCAT{\Tuple{w_1,x_1}}{}{t}{\Tuple{w_2,x_2}}\in\BehS{}}},
\end{equation}
 and $\ES{}$ is a \textbf{state space representation} of $\E=\Tuple{\Nbn,\WT,\Beh}$ if ${\projState{\WT}{\BehS{}}=\Beh}$.
Recalling the gedankenexperiment in Section~\ref{sec:lcom}, all necessary information to determine the future evolution (i.e., the next feasible domino) is captured in the last $l$ symbols.
Systems which exhibit this property are said to have \textbf{memory span $l$} \citep[p.184]{Willems1989}, formally
\begin{equation}\label{equ:def:memoryspan}
 \AllQ{w_1,w_2\in\Beh,t\in\Nbn}{\propImp*{w_1\ll{t,t+l-1}=w_2\ll{t,t+l-1}}{\CONCAT{w_1}{}{t}{w_2}\in\Beh}.}
\end{equation}

From \eqref{equ:def:memoryspan} we can conclude that 
\begin{inparaenum}[(i)]
\item every $l$-com\-plete system has memory span $l$,
 \item the state property implies that $\E_x=(\Nbn,\allowbreak	\X,\allowbreak\projState{\X}{\BehS{}})$ has memory span one and 
 \item a straightforward choice for the state space of an $l$-complete system is given by the set of admissible strings\footnote{In contrast to \citep[p.6]{MoorRaisch1999} this choice of the state space represents only the reachable part of $W^{[0,l-1]}$.} of length $l$.
\end{inparaenum}
Considering also the fact that for the first $l$ time steps we can only memorize the symbols already seen, we can generalize the construction of a state space representation
given in \citep[p.6]{MoorRaisch1999} to $l$-complete dynamical systems as defined in \eqref{equ:def:lcompl}.
\begin{lemma}\label{lem:CorrPastIndSS}
Let ${\E=\Tuple{\Nbn,\WT,\Beh}}$ be an $l$-complete dynamical system. Furthermore, let 
\begin{equation*}
 \textstyle\X:=\UNION{\BR{\bigcup_{r\in[0,l-1]}\Beh\ll{0,r-1}}}{\BR{\bigcup_{t\in\Nbn}\Beh\ll{t,t+l-1}}}
\end{equation*}
and let $\BehS{}\subseteq(\W\times\X)^{\Nbn}$ \SUCHTHAT $\Tuple{\w,\x}\in\BehS{}$ \IFF
\begin{equation}\label{equ:lem:CorrPastIndSS}
 \x(t)=\begin{cases}
                                       \w\ll{0,t-1}&0\leq t<l\\\w\ll{t-l,t-1}&t\geq l
                                      \end{cases}
\end{equation}
and $w\in\Beh$.
 Then
 $\ES{}=\Tuple{\Nbn,\W,\X,\BehS{}}$ is a state space representation of $\E$.
\end{lemma}
\begin{proof}
${\projState{\WT}{\BehS{}}=\Beh}$ holds by construction. To show \eqref{equ:StateSpaceDynamicalSystem:1}, pick $\Tuple{w_1,x_1},\Tuple{w_2,x_2}\in\BehS{}$ and $t'\in\Nbn$ s.t. $x_1(t')=x_2(t')$ and show $\Tuple{w,x}=\CONCAT{\Tuple{w_1,x_1}}{}{t'}{\Tuple{w_2,x_2}}\in\BehS{}$:
observe that
\begin{align}\label{equ:proof:lem:CorrPastIndSS}
 x_1(t')&=w_1\ll{\max\{0,t'-l\},t'-1}\notag\\
 &=w_2\ll{\max\{0,t'-l\},t'-1}=x_2(t').
\end{align}
This implies for $t'<l$ that $w=\CONCAT{w_1}{}{t'}{w_2}=w_2\in\Beh$. From $\E$ being $l$-complete, it follows  that $\E$ has memory span $l$ and therefore\footnote{Observe that, under the premises of \eqref{equ:def:memoryspan}, $\CONCAT{w_1}{}{t}{w_2}=\CONCAT{w_1}{}{t+l}{w_2}$ in the right side of the implication in \eqref{equ:def:memoryspan}.} \eqref{equ:proof:lem:CorrPastIndSS} implies $w=\CONCAT{w_1}{}{t'}{w_2}\in\Beh$ for $t'\geq l$.  
Now remember that \eqref{equ:lem:CorrPastIndSS} holds for $\Tuple{w_1,x_1},\Tuple{w_2,x_2}\in\BehS{}$. Therefore, ${x=\CONCAT{x_1}{}{t'}{x_2}}$ implies that for all $t\in\Nbn$
\begin{equation*}
x(t)=\begin{cases}
\w_1\ll{0,t-1}&\propConj{\BR{t\leq t'}}{\BR{t<l}}\\
\w_1\ll{t-l,t-1}&\propConj{\BR{t\leq t'}}{\BR{t\geq l}}\\
\w_1\ll{0,t'-1}\sconc\w_2\ll{t',t-1}&\propConj{\BR{t'<t<t'+l}}{\BR{t<l}}\\
\w_1\ll{t-l,t'-1}\sconc\w_2\ll{t',t-1}&\propConj{\BR{t'<t<t'+l}}{\BR{t\geq l}}\\
\w_2\ll{t-l,t-1}&\BR{t\geq t'+l}.
\end{cases}
\end{equation*}
Hence, with $w=\CONCAT{w_1}{}{t'}{w_2}\in\Beh$, ${x=\CONCAT{x_1}{}{t'}{x_2}}$ satisfies \eqref{equ:lem:CorrPastIndSS}, proving $\Tuple{w,x}\in\BehS{}$. 
\end{proof}
Since the strongest $l$-complete approximation ${\ElMax{}=}{\allowbreak\Tuple{\Nbn,\allowbreak\WT,\allowbreak\BehlMax{}}}$ of any dynamical system $\E=\Tuple{\Nbn,\WT,\Beh}$ is $l$-complete, we can use \REFlem{lem:CorrPastIndSS} to construct a state space representation of $\ElMax{}$, denoted by $\ElMaxS{}=\Tuple{\Nbn,\W,\X,\BehlMaxS{}}$. Note that the state space $\X$ constructed in \REFlem{lem:CorrPastIndSS} has finitely many elements if $\length{\W}<\infty$. 
\begin{example}\label{exp:statespace}
\normalfont
 Recall that the system in Example~\ref{exp:lcomplete} is $2$-complete
 and \eqref{equ:exp:lcomplete:1} implies
 $\bigcup_{t\in\Nbn}\Beh\ll{t,t+1}=\Set{aa,ab,ba}$. Adding the set $\bigcup_{r\in[0,1]}\Beh\ll{0,r-1}=\Set{\lambda,a,b}$, the state space defined in \REFlem{lem:CorrPastIndSS} for a state space representation of the system $\E$ in \eqref{equ:exp:lcomplete} (and its strongest $2$-complete approximation) is given by 
 $X=\Set{\lambda,a,b,aa,ab,ba}$. Analogously, the state space representation of the strongest $1$-complete approximation of $\E$ has state space ${X=\Set{\lambda,a,b}}$.
\end{example}

\section{Finite State Machine Representations}\label{sec:FSMRepresentations}
Using the notation from \citep[Def.3]{MoorRaisch1999}, a \textbf{finite state machine} is a tuple $\Psys=(\X,\W,\tr,\Xo{})$, where $\X$ (with $\length{\X}<\infty)$ is the state space, $\W$ (with $\length{\W}<\infty)$ is the signal space, $\Xo{}\subseteq\X$ is the set of initial states and $\tr\subseteq\X\times\W\times\X$ is a next state relation.
Furthermore, the \textbf{full behavior induced by $\Psys$} is defined as
\begin{equation}\label{equ:Behtr}
 \Beh_f(\Psys)=\SetCompX{\Tuple{\w,\x}}{
\begin{propConjA}
x(0)\in\Xo{}\\
 \AllQ{t\in\Nbn}{\Tuple{\x(t),\w(t),\x(t+1)}\in\tr}
\end{propConjA}},
\end{equation}
and we say that ${\Psys=(\X,\W,\tr,\Xo{})}$ realizes $\ES{}=(\Nbn,\W, \allowbreak\X, \allowbreak\BehS{})$ if $\Beh_f(\Psys)=\BehS{}$.
\\
Recall that 
in the presented domino game, a transition from one state to another is represented by adding an allowed domino. However, the set of allowed dominos is time dependent since we have to pick from the subset $\Beh\ll{t,t+l}$ of all dominos $\Ds{l+1}$ at time $t$. This suggests that the next state relation of an $l$-complete system is generally time dependent. Therefore, we define a \textbf{time dependent finite state machine} (tFSM) $\Qsys=(\X,\W,\trt,\Xo{})$, where $\X$, $\W$ and $\Xo{}$ are defined as for an FSM and $\trt:\Nbn\rightarrow\twoup{\X\times\W\times\X}$ is a time dependent next state relation. Furthermore, we define the full behavior induced by $\Qsys$ analogously to \eqref{equ:Behtr} by
\begin{equation}\label{equ:Behtrt}
 \Beh_f(\Qsys)=\SetCompX{\Tuple{\w,\x}}{
\begin{propConjA}
x(0)\in\Xo{}\\
 \AllQ{t\in\Nbn}{\Tuple{\x(t),\w(t),\x(t+1)}\in\trt(t)}
\end{propConjA}	}
\end{equation}
and say that ${\Qsys=(\X,\W,\trt,\Xo{})}$ is realizing $\ES{}=(\Nbn,\W,\allowbreak\X,\allowbreak\BehS{})$ if $\Beh_f(\Qsys)=\BehS{}$.
Using the above intuition, we can show that this tFSM can be used to realize the $l$-complete state space system constructed in \REFlem{lem:CorrPastIndSS}. This extends \cite[Thm.12]{MoorRaisch1999} to $l$-complete dynamical systems in the sense of  \eqref{equ:def:lcompl}, including also time variant systems.
\begin{lemma}\label{lem:SforLcomplete}
Let ${\E=\Tuple{\Nbn,\WT,\Beh}}$ be an $l$-complete dynamical system with $\length{W}<\infty$ and $\ES{}=\Tuple{\Nbn,\W,\X,\BehS{}}$ its state space representation constructed in \REFlem{lem:CorrPastIndSS}. Then $\ES{}$ is realized by $Q=(\X,\W,\trt,\Xo{})$ with $\Xo{}=\Set{\lambda}$ and
 \begin{align}\label{equ:delta_timeV}
\trt(t)=&\SetCompX{\Tuple{\xi,\omega,\xi\sconc\omega}}{\propConj{t<l}{\xi\sconc\omega\in\Beh\ll{0,t}}}\\
         &\cup\SetCompX{\Tuple{\xi,\omega,\xi\lb{1,l-1}\sconc\omega}}{\propConj{t\geq l}{\xi\sconc\omega\in\Beh\ll{t-l,t}}}.\notag
\end{align}
\end{lemma}

\begin{proof}
Finiteness of $X$ follows from  $\length{W}<\infty$ and the construction of $X$ in \REFlem{lem:CorrPastIndSS}. Observe that using \eqref{equ:delta_timeV} in \eqref{equ:Behtrt} gives 
\begin{equation}\label{equ:proof:lem:SforLcomplete}
 \Beh_f(\Qsys)=\SetCompX{\Tuple{\w,\x}}{
\AllQ{t\in\Nbn}{
\begin{propConjA}
 \x(t)=\w\ll{\max\{0,t-l\},t-1}\\
\SplitS{\w\ll{\max\{0,t-l\},t-1}}{\in\Beh\ll{\max\{0,t-l\},t-1}}
\end{propConjA}	}
}
\end{equation}
as the induced full behavior.
Now $l$-completeness of $\E$ and the last line in \eqref{equ:proof:lem:SforLcomplete} imply $w\in\Beh$. Furthermore, \eqref{equ:proof:lem:SforLcomplete} immediately implies that \eqref{equ:lem:CorrPastIndSS} holds, which gives $\Beh_f(\Qsys)=\BehS{}$ from the construction of $\BehS{}$ in \REFlem{lem:CorrPastIndSS}.
\end{proof}

\begin{remark}\label{rem:easierPsi}
Recall the gedankenexperiment in Section~\ref{sec:lcom} and observe that in the construction of \REFlem{lem:CorrPastIndSS} the state represents the \enquote{recent past} of the signal $w$, i.e., a finite string of length $l$ if $t\geq l$.
However, at start up, i.e., for $t<l$, no \enquote{past} of this length exists. Then the state describes the  available past information, i.e., a finite string of length $r\in[0,l-1]$ contained in the set $\Beh\ll{0,r-1}$.
Therefore, assuming $\xi=\Tuple{\omega_0,\hdots,\omega_{r-1}}$ s.t. $\lengthw{\xi}=r<l$ implies that $\Tuple{\xi,\omega,\xi'}\in\trt(t)$ \IFF $\xi'=\xi\sconc\w=\Tuple{\omega_0,\hdots,\omega_{r-1},\omega}$ is the extension of $\xi$ by $\omega$ and a valid initial behavior, i.e., $\xi'\in\Beh\ll{0,r}$.
Now remember that the domino game describes the admissible behavior by appending domino pie\-ces of length $l+1$ such that the last $l$ symbols match. Therefore, assuming $\xi=\Tuple{\omega_0,\hdots,\omega_{l-1}}$ (i.e., $\lengthw{\xi}=l$) implies  ${\Tuple{\xi,\omega,\xi'}\in\trt(t)}$ \IFF  $\xi'=\xi\lb{1,l-1}\sconc\omega=\Tuple{\omega_1,\hdots,\omega_{l-1},\omega}$ and ${\xi\sconc\omega=\Tuple{\omega_0,\hdots,\omega_{l-1},\omega}\in\Beh\ll{t-l,t}}$, i.e., ${\xi\sconc\omega}$ is a domino that is currently allowed to be  attached. For an illustration of the last case, see Figure~\ref{fig:CorrespondanceDominoGameTFSM}.
\end{remark}

\begin{figure}[htb!]
\begin{center}
 \begin{tikzpicture}[auto,scale=1]
 \begin{pgfonlayer}{background}
   \path      (0,0) node (o) {
      \includegraphics[width=0.215\linewidth]{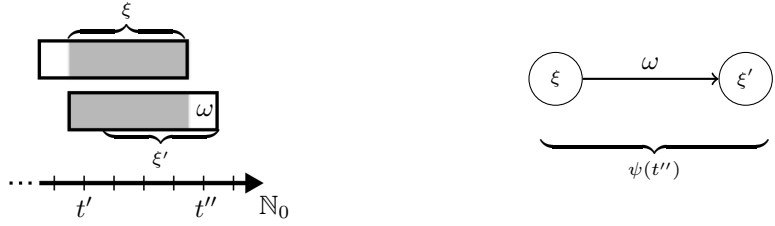}};
 \end{pgfonlayer}
 \begin{pgfonlayer}{foreground}
  \path (o.north west)+(1.5,-0.45) node (i) {};
    \path (o.north west)+(1.75,-0) node {$\overbrace{\phantom{\,a\,a\,a\,a\,a\,a\,}}^{ \xi}$};
  \path (i)+(0.7,-1.1) node (ii){$\underbrace{\phantom{\,a\,a\,a\,a\,a\,a\,}}_{\xi'}$};
  \path (ii.north east)+(-0.3,-0.05) node (iii){$\omega$};
  \path (o.south west)+(1.2,0) node (t2) {$t'$};
  \path (t2)+(1.6,0) node (t6) {$t''$};
  \path (t6)+(0.9,0) node (t6) {$\Nbn$};
 \end{pgfonlayer}
    \begin{scriptsize} 
\node[state] (a) at (5.5,0.5) {$\xi$};
\node[state] (b) at (8,0.5) {$\xi'$};
\end{scriptsize} 
\SFSAutomatEdge{a}{\omega}{b}{}{}
\node (brace) at (6.8,-0.5) {$\underbrace{\phantom{\,a\,a\,a\,a\,a\,a\,a\,a\,a\,a\,a\,a\,}}_{\psi(t'')}$};
  \end{tikzpicture}
    \end{center}
\vspace{-0.5cm}
  \caption{Correspondence of one step in the domino game (left) to one time dependent transition in a tFSM (right), where $t'=t''-l$ and $l=4$.}\label{fig:CorrespondanceDominoGameTFSM}
\end{figure}

Consider the strongest $l$-complete approximation $\ElMax{}=\Tuple{\Nbn,\WT,\BehlMax{}}$ of any dynamical system $\E=\Tuple{\Nbn,\WT,\Beh}$. Then the state space representation $\ElMaxS{}=\Tuple{\Nbn,\W,\X,\BehlMaxS{}}$ suggested in \REFlem{lem:CorrPastIndSS} can obviously be realized by the tFSM $\Qsys$ in \REFlem{lem:SforLcomplete}.

\begin{example}\label{exp:FSM}
\normalfont
 Using the state spaces derived in Example~\ref{exp:statespace} and the construction of the time dependent next state relation in \eqref{equ:delta_timeV}, we can construct the tFSMs $\Qsys^{1}$ and $\Qsys^{2}$, depicted in Figure~\ref{fig:FSMtype}, realizing the strongest $1$- and $2$-complete approximations of the system $\E$ in \eqref{equ:exp:lcomplete}, respectively. As $\E$ is $2$-complete, the tFSM $\Qsys^{2}$ is also a realization of $\E$. Observe that the tFSM $\Qsys^{1}$ reduces to a standard FSM (due to \eqref{equ:exp:lcomplete:1}). In $\Qsys^{2}$ the transition from state $aa$ to itself is time dependent, because three sequential $a$'s are only allowed at start up. In both figures the initial state is indicated by an arrow pointing to it from \enquote{outside}.
 \end{example}
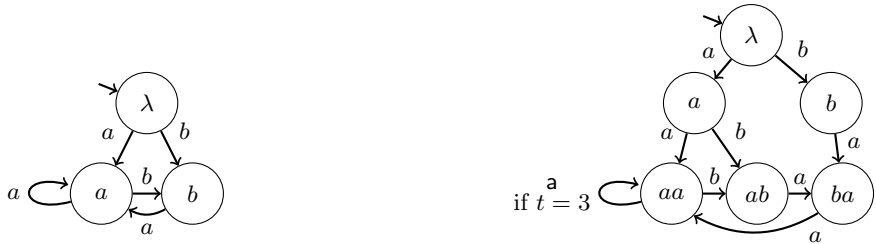
\begin{figure}[htb]
\begin{center}
 \begin{tikzpicture}[auto,scale=1.5]
    \begin{small} 
\node (dummy) at (-0.5,-0) {};  
\node[state] (p0) at (0,-0.2) {$\lambda$};
\node[state] (a) at (-0.4,-1) {$a$};
\node[state] (b) at (0.4,-1) {$b$};

\SFSAutomatEdge{dummy}{}{p0}{}{}
\SFSAutomatEdge{p0}{a}{a}{}{swap}
\SFSAutomatEdge{p0}{b}{b}{}{}
\SFSAutomatEdge{a}{a}{a}{loop left}{}
\SFSAutomatEdge{a}{b}{b}{}{}
\SFSAutomatEdge{b}{a}{a}{bend left}{}

\end{small} 
    \begin{small} 
\node (dummy) at (4.8,0.6) {};  
\node[state] (p0) at (5.3,0.4) {$\lambda$};
\node[state] (a) at (4.8,-0.2) {$a$};
\node[state] (b) at (6,-0.2) {$b$};
\node[state] (aa) at (4.6,-1) {$aa$};
\node[state] (ab) at (5.35,-1) {$ab$};
\node[state] (ba) at (6.1,-1) {$ba$};
\SFSAutomatEdge{dummy}{}{p0}{}{}
\SFSAutomatEdge{p0}{a}{a}{}{swap,pos=0.4}
\SFSAutomatEdge{p0}{b}{b}{}{pos=0.4}
\SFSAutomatEdge{a}{a}{aa}{}{swap}
\SFSAutomatEdge{a}{b}{ab}{}{}
\SFSAutomatEdge{b}{a}{ba}{}{pos=0.8}
\SFSAutomatEdge{aa}{b}{ab}{}{}
\SFSAutomatEdge{ab}{a}{ba}{}{}
\SFSAutomatEdge{ba}{a}{aa.south east}{bend left}{pos=0.15}
\SFSAutomatEdge{aa}{\begin{matrix}\mathsf{a}\\[-1mm]\text{if}~t=3\end{matrix}}{aa}{loop left}{}
\end{small} 
\end{tikzpicture}
\end{center}
 \vspace{-0.3cm}
 \caption{tFSMs $\Qsys^{1}$ (left) and $\Qsys^{2}$ (right) constructed in Example~\ref{exp:FSM}, realizing $\Sigma^{1^\Uparrow}_{S}$  and $\Sigma^{2^\Uparrow}_{S}$  of \eqref{equ:exp:lcomplete}, respectively.}\label{fig:FSMtype}
 \end{figure}

\section{Asynchronous Properties}\label{sec:AsyncProp}

Obviously, one could render the next state relation in \eqref{equ:delta_timeV} time independent by using time as an additional state variable. However, this would lead to an infinite state set. We want to characterize systems naturally allowing an FSM realization, (i.e., a time independent next state relation). Observe that such systems must allow for concatenation of state trajectories that reach the same state asynchronously (i.e., at different times). This is formalized in the following definition inspired by \cite[p.59]{JuliusPhdThesis2005}.
\begin{definition}\label{def:StateSpaceDynamicalSystemTA}
Let $\ES{}=\Tuple{\Nbn,\WT,\X,\BehS{}}$ be a dynamical system with internal signal space $X$.
Then $\ES{}$ is an \textbf{asynchronous state space dynamical system} if
\begin{equation}\label{equ:StateSpaceDynamicalSystem:2}
 \AllQ{\Tuple{w_1,x_1},\Tuple{w_2,x_2}\in\BehS{}, t_1,t_2\in\Nbn}{\propImp*{x_1(t_1)=x_2(t_2)}{\CONCAT{\Tuple{w_1,x_1}}{t_1}{t_2}{\Tuple{w_2,x_2}}\in\BehS{}}}.
\end{equation}
\end{definition}
It can be easily observed that every asynchronous state space dynamical system is also a synchronous\footnote{To clearly 
distinguish the asynchronous state property and the (standard) state property from Section~\ref{sec:statespace}, we will in the remainder of this paper refer to the latter one as \textbf{synchronous} state property. The same convention is applied to other properties as memory span and $l$-completeness.}
state space dynamical system since we can always pick $t_1=t_2=t$ in \eqref{equ:StateSpaceDynamicalSystem:2} and get \eqref{equ:StateSpaceDynamicalSystem:1}.\\
It is important to understand that the asynchronous state property does not imply and is not implied by the time invariance property of dynamical systems, since it depends on the realization of the system. 
This is illustrated by the following example.
\begin{example}\label{exp:FSMvsFSMtype}\normalfont
The FSM $P=\Tuple{\Set{\xi_1,\xi_2},\Set{a,b},\delta,\Set{\xi_1}}$ in Figure~\ref{fig:rem:FSMvsFSMtype:2} (left) realizes the asynchronous state space system $\ES{a}=\Tuple{\Nbn,\Set{a,b},\Set{\xi_1,\xi_2},\Beh_f(P)}$ representing the time variant system $\E_a=\Tuple{\Nbn,\Set{a,b},\Set{ab^\omega}}$.  Furthermore, the tFSM $Q=\Tuple{\Set{\xi},\Set{a,b},\psi,\Set{\xi}}$ in Figure~\ref{fig:rem:FSMvsFSMtype:2} (right) realizes the synchronous (but not asynchronous) state space system $\ES{b}=\Tuple{\Nbn,\Set{a,b},\Set{\xi},\Beh_f(Q)}$ representing the time invariant system $\E_b=\Tuple{\Nbn,{a,b},\Set{ab^\omega,b^\omega}}$.
\end{example}
\begin{figure}[htb!]
 \begin{center}
  \begin{tikzpicture}[auto,scale=1.3]
   \begin{scriptsize} 
\node (dummy) at (-0.6,0) {}; 
\node[state] (a) at (0,0) {$\xi_1$};
\node[state] (b) at (1,0) {$\xi_2$};
\end{scriptsize} 
\SFSAutomatEdge{dummy}{}{a}{}{}
\SFSAutomatEdge{a}{a}{b}{}{}
\SFSAutomatEdge{b}{b}{b}{loop right}{}
   \begin{scriptsize} 
   \node (dummy) at (5,0) {}; 
\node[state] (a) at (5.6,0) {$\xi$};
\end{scriptsize} 
\SFSAutomatEdge{dummy}{}{a}{}{}
\SFSAutomatEdge{a}{\begin{matrix}
                    a~\text{if}~t=0\\[0.1cm]
                    b~\forall t\in\Nbn
                   \end{matrix}
}{a}{loop right}{}
\end{tikzpicture}
\end{center}
\vspace{-0.3cm}
\caption{FSM $P$ (left) and tFSM $Q$ (right) in Example~\ref{exp:FSMvsFSMtype}.}\label{fig:rem:FSMvsFSMtype:2}
 \end{figure}
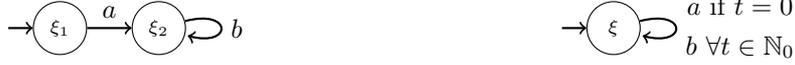

Recall that the concepts of synchronous state property and synchronous memory span are strongly related, since the synchronous state property implies that $\E_x=\Tuple{\Nbn,\X,\projState{\X}{\BehS{}}}$ has memory span one. To get the same relation for the asynchronous case, we define an asynchronous memory span.
\begin{definition}\label{def:memoryspan_TA}
The dynamical system  ${\E=\Tuple{\Nbn,\WT,\Beh}}$ has \textbf{asynchronous memory span $l$} if
\begin{equation}\label{equ:memoryspan:1}
 \AllQ{w_1,w_2\in\Beh,t_1,t_2\in\Nbn}{\propImp*{w_1\ll{t_1,t_1+l-1}=w_2\ll{t_2,t_2+l-1}}{\CONCAT{w_1}{t_1}{t_2}{w_2}\in\Beh}.}
\end{equation}
\end{definition}
As expected, it can be easily seen that every system with asynchronous memory span $l$ also has synchronous memory span $l$.\\
For systems with an asynchronous memory span, the do\-mi\-no game presented in \REFsec{sec:lcom} is significantly simplified. At any time $t$ we can attach any domino from the whole domino set $\Ds{l+1}=\bigcup_{t'\in\Nbn}\Beh\ll{t',t'+l}$, as long as the first $l$ symbols of the newly attached domino match the last $l$ symbols of the previous domino. Recall that this implies time independent transitions in a corresponding FSM realization, which is what we are aiming at. 
Having this interpretation in mind, the definition of asynchronous $l$-completeness comes as no surprise. 

\begin{definition}\label{def:lcnessSC_TA}
The system ${\E=\Tuple{\Nbn,\WT,\Beh}}$ is \textbf{asynchro\-nously $l$-complete} if
\begin{equation}\label{equ:lcnessSC:1}
\propAequ{
\begin{propConjA}
w\ll{0,l}\in\Beh\ll{0,l}\\
 \AllQ{t\in\Nbn}{w\ll{t,t+l}\in \bigcup_{t'\in\Nbn}\Beh\ll{t',t'+l}}
\end{propConjA}}{w\in\Beh}.
\end{equation}
\end{definition}

Again, it is easily verified that a system is synchronously $l$-complete if it is asynchronously $l$-complete.

\begin{remark}\label{rem:firstLine}
The second line in \eqref{equ:lcnessSC:1} describes that the possible future evolution of the system depends on the $l$ past values of a signal if $t\geq l$. However, at start up this \enquote{past} is not yet fully available. Therefore, the first line in \eqref{equ:lcnessSC:1} is needed to ensure that all signals start with an allowed initial pattern. However, observe that if $\E$ is time invariant, the condition $\sigma\Beh\subseteq\Beh$ implies $\AllQ{t\in\Nbn}{\sigma^t\Beh\ll{t,t+l}\subseteq\Beh\ll{0,l}}$ giving $\bigcup_{t'\in\Nbn}\Beh\ll{t',t'+l}=\Beh\ll{0,l}$. Then the first line in \eqref{equ:lcnessSC:1} is implied by the second line and is therefore unnecessary. This is stated in the following lemma.
\end{remark}

\begin{lemma}\label{lem:lcnessTI_TA}
Let $\E=\Tuple{\Nbn,\WT,\Beh}$ be a time invariant dynamical system.  Then $\E$ is asynchronously $l$-complete \IFF 
\begin{equation}\label{equ:lem:lcnessTI_TA}
\propAequ{\AllQ*{t\in\Nbn}{w\ll{t,t+l}\in \Beh\ll{0,l}}}{w\in\Beh}.
\end{equation}
\end{lemma}

\begin{proof}
As pointed out in Remark~\ref{rem:firstLine}, time invariance of $\E$ implies $\bigcup_{t'\in\Nbn}\Beh\ll{t',t'+l}=\Beh\ll{0,l}$. Hence \eqref{equ:lem:lcnessTI_TA} and \eqref{equ:lcnessSC:1} are identical.
\end{proof}

\begin{remark}
 Remember that a synchronously $l$-complete system always has synchronous memory span $l$. However, the reverse implication only holds if the system is complete to ensure that its behavior can be fully described by a local property such as a finite memory span. This statement was proven in \citep[prop.1.1]{Willems1989} for the synchronous case and can be generalized to the asynchronous case, where the proof follows the same lines. This emphasizes that the asynchronous properties extend the behavioral systems theory in a consistent way.
\end{remark}

\begin{remark}\label{rem:AnmerkungMoorRaisch:2}
 Recall from Remark~\ref{rem:AnmerkungMoorRaisch} that in \cite{MoorRaisch1999} $l$-complete\-ness for time invariant systems is defined by \eqref{equ:lem:lcnessTI_TA} (instead of \eqref{equ:def:lcompl}). Therefore, \REFlem{lem:lcnessTI_TA} implies that this weaker version of $l$-completeness from \cite{MoorRaisch1999} coincides with the property of asynchronous $l$-completeness for time-invariant systems. 
\end{remark}

\begin{example}\label{exp:alcomplete}
\normalfont
 We now investigate the asynchronous $l$-com\-pleteness properties of the system $\E$ in \eqref{equ:exp:lcomplete}. Since $\E$ is time invariant, it follows from Remark~\ref{rem:firstLine} that $\bigcup_{t'\in\Nbn}\Beh\ll{t',t'+l}=\Beh\ll{0,l}$. Therefore, the simplified domino game for $l=1$ is identical to the one played in Example~\ref{exp:lcomplete}, implying that the system \eqref{equ:exp:lcomplete} is not asynchronously $1$-complete. For $l=2$, observe that in the simplified domino game we are still allowed to use the piece $aaa$ from the set $\Beh\ll{0,2}$ at any time $t>0$. Therefore, more than two sequential $a$'s can be produced by this game implying that the system \eqref{equ:exp:lcomplete} is \emph{not} asynchronously $2$-complete. Extending $l$ to $l=3$ gives the domino set
 $\bigcup_{t'\in\Nbn}\Beh\ll{t',t'+3}=\Beh\ll{0,3}=\Set{aaab,aaba,abaa,baab}$.
 Now, playing the simplified domino game ensures that always three symbols have to match, preventing the piece $aaab$ to be attachable for $t>0$. Hence, the resulting behavior is identical to $\Beh$. This implies that the system \eqref{equ:exp:lcomplete} is asynchronously $3$-complete.
\end{example}

If we recall that the memory of the system is still given by the last $l$ symbols of the signal $w$ it is obvious that we can construct a state space representation of an asynchronously $l$-complete approximation exactly as given in \REFlem{lem:CorrPastIndSS} for the synchronous case. However, dealing with the asynchronous version, we can realize it by an FSM.

\begin{lemma}\label{lem:realizationTA}
Let ${\E=\Tuple{\Nbn,\WT,\Beh}}$ be an asynchronously $l$-com\-plete dynamical system. 
Then $\ES{}=\Tuple{\Nbn,\W,\X,\BehS{}}$ from \REFlem{lem:CorrPastIndSS} is an asynchronous state space representation of $\E$.
Furthermore, if $\length{W}<\infty$, $\ES{}$ is realized by the finite state machine $\Psys=(\X,\W,\tr,\Xo{})$ with $\Xo{}=\Set{\lambda}$ and 
 \begin{align}\label{equ:delta_timeInv}
\tr= &\SetCompX{\Tuple{\xi,\omega,\xi\sconc\omega}}{\propConj{\lengthw{\xi}<l}{\xi\sconc\omega\in\Beh\ll{0,\lengthw{\xi}}}}\\
         &\cup\SetCompX{\Tuple{\xi,\omega,\xi\lb{1,l-1}\sconc\omega}}{\propConj{\lengthw{\xi}=l}{\textstyle\xi\sconc\omega\in\bigcup_{t'\in\Nbn}\Beh\ll{t',t'+l}}}.\notag
\end{align}

\end{lemma}
\begin{proof}
The proof follows the same lines as the proofs of \REFlem{lem:CorrPastIndSS} and \REFlem{lem:SforLcomplete} and is therefore omitted.
\end{proof}

\begin{remark}
 The next state relation $\delta$ in \eqref{equ:delta_timeInv} can be interpreted analogously to $\psi$ in \eqref{equ:delta_timeV}, see the discussion in Remark~\ref{rem:easierPsi}. Observe that now the condition in the last line of \eqref{equ:delta_timeInv} is weakened in the sense that ${\xi\sconc\omega}$ can be any domino in the gedankenexperiment in Section~\ref{sec:lcom}.
\end{remark}

\begin{example}\label{exp:astatespace}
\normalfont
 Recall from Example~\ref{exp:alcomplete} that the system \eqref{equ:exp:lcomplete} in Example~\ref{exp:lcomplete} is asynchronously $3$-complete
 and that \eqref{equ:exp:lcomplete:2} implies
 $\bigcup_{t\in\Nbn}\Beh\ll{t,t+2}=\Set{aaa,aab,aba,baa}$. Adding the set $\bigcup_{r\in[0,2]}\Beh\ll{0,r-1}=\Set{\lambda,a,b,aa,ab,ba}$, the state space defined in \REFlem{lem:CorrPastIndSS} with $l=3$ for the system in \eqref{equ:exp:lcomplete} is given by $X=\{\lambda,a,b,aa,ab,ba,aaa,aab,\allowbreak aba,baa\}$.
 Using this state space and the construction of the next state relation in \eqref{equ:delta_timeInv}, we can construct an FSM $\Psys$ realizing the system \eqref{equ:exp:lcomplete} in Example~\ref{exp:lcomplete}. The result is depicted in Figure~\ref{fig:FSM}.
 \end{example}
\begin{figure}[htb]
 \begin{center}
  \begin{tikzpicture}[auto,scale=1.5]
    \begin{small} 
\node (dummy) at (-0.5,-0) {}; 
\node[state] (p0) at (0,-0.2) {$\lambda$};
\node[state] (a) at (-1,-0.4) {$a$};
\node[state] (b) at (1,-0.4) {$b$};
\node[state] (aa) at (-1.5,-1) {$aa$};
\node[state] (ab) at (-0.5,-1) {$ab$};
\node[state] (ba) at (0.5,-1) {$ba$};
\node[state] (aaa) at (-2,-1.6) {$aaa$};
\node[state] (aab) at (-1,-1.6) {$aab$};
\node[state] (aba) at (0,-1.6) {$aba$};
\node[state] (baa) at (1,-1.6) {$baa$};

\SFSAutomatEdge{dummy}{}{p0}{}{}
\SFSAutomatEdge{p0}{a}{a}{}{swap}
\SFSAutomatEdge{p0}{b}{b}{}{}
\SFSAutomatEdge{a}{a}{aa}{}{swap}
\SFSAutomatEdge{a}{b}{ab}{}{}
\SFSAutomatEdge{b}{a}{ba}{}{swap}
\SFSAutomatEdge{aa}{a}{aaa}{}{swap}
\SFSAutomatEdge{aa}{b}{aab}{}{}
\SFSAutomatEdge{ab}{a}{aba}{}{}
\SFSAutomatEdge{ba}{a}{baa}{}{}

\SFSAutomatEdge{aaa}{b}{aab}{}{}
\SFSAutomatEdge{aab}{a}{aba}{}{}
\SFSAutomatEdge{aba}{a}{baa}{}{}
\SFSAutomatEdge{baa}{b}{aab.south east}{bend left}{pos=0.8}
\end{small}
\end{tikzpicture}
\end{center}
\vspace{-0.7cm}
\caption{FSM $\Psys$ realizing  \eqref{equ:exp:lcomplete}.}\label{fig:FSM}
 \end{figure}
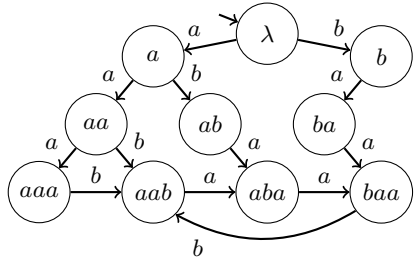

\begin{remark}\label{rem:STIversusTI}
Observe that \eqref{equ:lem:lcnessTI:neu} in \REFlem{lem:lcnessTI} and \eqref{equ:lem:lcnessTI_TA} in \REFlem{lem:lcnessTI_TA} are identical. Therefore, \REFlem{lem:lcnessTI} and \ref{lem:lcnessTI_TA} imply that the asynchronous and the synchronous $l$-com\-pleteness property coincide for strictly time invariant systems. 
As a direct consequence, the state space representation of a strictly time invariant (synchronously) $l$-complete system can be realized by the FSM $\Psys$ constructed in \REFlem{lem:realizationTA}.
\end{remark}

\section{Asynchronous $l$-Complete Approximation}\label{sec:alcomplapprox}

Using the asynchronous $l$-completeness property introduced in \REFdef{def:lcnessSC_TA}, we can construct asynchronous $l$-com\-plete approximations analogously to their synchronous versions in \REFsec{sec:lcom}. 

\begin{definition}\label{def:AsyncLcompApprox}
Let $\E{}=\Tuple{\Nbn,\W,\Beh}$ be a dynamical system.
Then $\El{}=\Tuple{\Nbn,\W,\Behl{}}$ is an \textbf{asynchronous $l$-complete approximation} of $\E$, if
\begin{inparaenum}[(i)]
 \item 
 $\El{}$ is asynchronously $l$-complete and 
 \item 
 $\Behl{}\supseteq\Beh{}$.\\
\end{inparaenum}
Furthermore, $\ElaMax{}=\Tuple{\Nbn,\W,\BehlaMax{}}$ is the \textbf{strongest asyn\-chro\-nous $l$-complete approximation} of $\E{}$, if 
\begin{inparaenum}[(i)]
 \item $\ElaMax{}$ is an asynchronous $l$-com\-plete approximation of $\E{}$ and
 \item for any asynchronous $l$-complete approximation $\E{}'=\Tuple{\Nbn,\W,\Beh{}'}$ of $\E{}$ it holds that $\BehlaMax{}\subseteq\Beh{}'$.
\end{inparaenum}
\end{definition}

Recall that for an asynchronously $l$-complete system, the domino game gedankenexperiment can be simplified such that at any time $t$ we can attach any domino from the whole domino set $\Ds{l+1}$. This simplified domino game now constructs the unique strongest asynchronous $l$-complete approximation $\BehlaMax{}$.

\begin{lemma}\label{lem:constructElMax_general_TA}
 Let $\E{}=\Tuple{\Nbn,\W,\Beh}$ be a dynamical system. Then the unique strongest asynchronous $l$-complete approximation of $\E{}$ is given by $\ElaMax{}=\Tuple{\Nbn,\W,\BehlaMax{}}$, with
 \begin{equation}\label{equ:lem:constructElMax_general_TA}
\BehlaMax{}:=\SetCompX{\w}{\begin{propConjA}
w\ll{0,l}\in\Beh\ll{0,l}\\
 \AllQ{t\in\Nbn}{w\ll{t,t+l}\in \bigcup_{t'\in\Nbn}\Beh\ll{t',t'+l}}
\end{propConjA}}.
 \end{equation}
  Furthermore, if $\E$ is time invariant then 
  \begin{equation}\label{equ:lem:constructElMax_general:3}
 \BehlaMax{}=\SetCompX{\w\in\W^{\Nbn}}{\AllQ{t\in \Nbn}{\w\ll{t,t+l}\in \Beh\ll{0,l}}}.
 \end{equation}
\end{lemma}

\begin{proof}
The proof of the first part follows the same lines as the proof of \REFlem{lem:constructElMax_general}, and the second part follows directly from \REFlem{lem:lcnessTI_TA}.
\end{proof}

As a direct consequence of \REFlem{lem:realizationTA}, the strongest asynchronous $l$-complete approximation $\ElaMax{}=\Tuple{\Nbn,\W,\BehlaMax{}}$ of any dynamical system $\E{}=\Tuple{\Nbn,\W,\Beh}$ can be represented by the state space system constructed in \REFlem{lem:CorrPastIndSS}, denoted by  $\ElaMaxS{}=\Tuple{\Nbn,\W,\X,\BehlaMaxS{}}$ which can be realized by the FSM $\Psys=\Tuple{\X,\W,\tr,\Xo{}}$ given in \REFlem{lem:realizationTA}.\\
Furthermore, as a direct consequence from Remark~\ref{rem:STIversusTI}, both approximation techniques coincide for strictly time invariant systems. 

\begin{lemma}\label{lem:lcnessTI_Approx}
Let $\E=\Tuple{\Nbn,\WT,\Beh}$ be a strictly time invariant dynamical system and $l\in\Nbn$.
Then its strongest synchronous $l$-complete approximation $\ElMax{}$ and its strongest asynchronous $l$-complete approximation $\ElaMax{}$ are identical, i.e., $\ElMax{}=\ElaMax{}$.
\end{lemma}

\begin{proof}
Follows directly from \REFlem{lem:constructElMax_general} and \REFlem{lem:constructElMax_general_TA}.
\end{proof}

\begin{remark}
 Recall from Remark~\ref{rem:AnmerkungMoorRaisch:2} that the weaker notion of $l$-completeness from \cite{MoorRaisch1999} coincides with the property of asynchronous $l$-completeness for time-invariant systems. Therefore, the strongest $l$-complete approximation of a time invariant system $\E$ suggested in \cite{MoorRaisch1999} is identical to its strong\-est asynchronous $l$-complete approximation $\ElaMax{}$ introduced in \REFdef{def:AsyncLcompApprox}. The latter is, by definition, also a synchronous $l$-complete approximation, but not (unless $\E$ is strictly time invariant) necessarily the strongest one.
\end{remark}

\begin{example}
\normalfont
The behaviors constructed by the domino games discussed in Example~\ref{exp:alcomplete} characterize the strongest asynchro\-nous $1$, $2$ and $3$-complete approximations for the system $\E$ in \eqref{equ:exp:lcomplete}, respectively. Realizations for the strongest asynchronous $1$- and $2$-complete approximations using the constructions from Lemma~\ref{lem:realizationTA} are shown in Figure~\ref{fig:FSM_12}. As the system $\E$ is asynchronously $3$-complete, its behavior coincides with that of its strongest $3$-complete approximation; hence the corresponding FSM is shown in Figure~\ref{fig:FSM}. Observe that the FSM realizing $\Sigma^{1^\uparrow}_{S}$ and the tFSM realizing $\Sigma^{1^\Uparrow}_{S}$ depicted in Figure~\ref{fig:FSM_12} (left) and Figure~\ref{fig:FSMtype} (left), respectively, coincide. This is a direct consequence from Remark~\ref{rem:STIversusTI}, since $\E^{1^\Uparrow}$ is strictly time invariant as discussed in Example~\ref{exp:Beh1strictlyTI}. 
 \end{example}

\begin{figure}[htb]
\begin{center}
 \begin{tikzpicture}[auto,scale=1.5]
    \begin{small} 
\node (dummy) at (-0.5,-0) {};  
\node[state] (p0) at (0,-0.2) {$\lambda$};
\node[state] (a) at (-0.4,-1) {$a$};
\node[state] (b) at (0.4,-1) {$b$};

\SFSAutomatEdge{dummy}{}{p0}{}{}
\SFSAutomatEdge{p0}{a}{a}{}{swap}
\SFSAutomatEdge{p0}{b}{b}{}{}
\SFSAutomatEdge{a}{a}{a}{loop left}{}
\SFSAutomatEdge{a}{b}{b}{}{}
\SFSAutomatEdge{b}{a}{a}{bend left}{}

\end{small} 
    \begin{small} 
\node (dummy) at (4.8,0.6) {};  
\node[state] (p0) at (5.3,0.4) {$\lambda$};
\node[state] (a) at (4.8,-0.2) {$a$};
\node[state] (b) at (6,-0.2) {$b$};
\node[state] (aa) at (4.6,-1) {$aa$};
\node[state] (ab) at (5.35,-1) {$ab$};
\node[state] (ba) at (6.1,-1) {$ba$};
\SFSAutomatEdge{dummy}{}{p0}{}{}
\SFSAutomatEdge{p0}{a}{a}{}{swap,pos=0.4}
\SFSAutomatEdge{p0}{b}{b}{}{pos=0.4}
\SFSAutomatEdge{a}{a}{aa}{}{swap}
\SFSAutomatEdge{a}{b}{ab}{}{}
\SFSAutomatEdge{b}{a}{ba}{}{pos=0.8}
\SFSAutomatEdge{aa}{b}{ab}{}{}
\SFSAutomatEdge{ab}{a}{ba}{}{}
\SFSAutomatEdge{ba}{a}{aa.south east}{bend left}{pos=0.15}
\SFSAutomatEdge{aa}{a}{aa}{loop left}{}
\end{small} 
\end{tikzpicture}
\end{center}
 \vspace{-0.3cm}
 \caption{FSMs realizing the strongest asynchronous $1$-complete approximation $\Sigma^{1^\uparrow}_{S}$ (left) and the strongest asynchronous $2$-complete approximation $\Sigma^{2^\uparrow}_{S}$ (right) of \eqref{equ:exp:lcomplete}.}\label{fig:FSM_12}
 \end{figure}
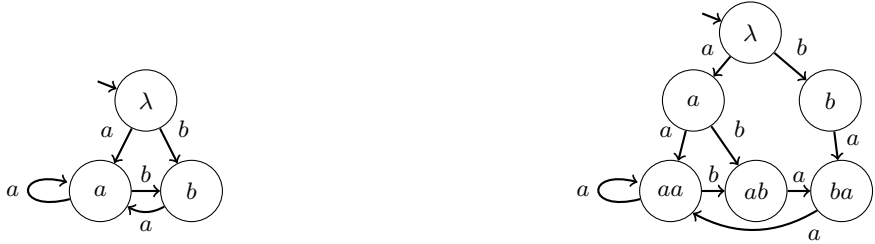

Summarizing the results of our running example, we have the following: the system under consideration, $\Sigma=(\Nbn,\allowbreak\W,\allowbreak\Beh)$ in \eqref{equ:exp:lcomplete}, is time invariant but not strictly time invariant. $\Sigma$ is synchronously $2$-complete and can therefore be realized by the tFSM $Q^2$ depicted in Figure~\ref{fig:FSMtype} (right). It is asynchronously $3$-complete (but not asynchronous $2$-complete as $\E$ is not strictly time invariant) and can therefore be realized by an FSM $P$ depicted in Figure~\ref{fig:FSM}. Its strongest asynchronous $2$-compete approximation $\Sigma^{2^\uparrow}$ is of cause also a synchronous $2$-complete approximation of $\E$, but not the strongest one. In fact, as $\Sigma$ is synchronously $2$-complete and asynchronously $3$-complete, $\Sigma^{2^\uparrow}=\Sigma^{3^\Uparrow}=\Sigma$, and therefore $\Beh^{2^\uparrow}=\Beh^{3^\Uparrow}=\Beh\subset\Beh^{2^\Uparrow}$.

\section{Conclusion}
Strongest $l$-complete approximations for time invariant systems were introduced in \cite{MoorRaisch1999}. However, the employed notion of $l$-completeness is a weaker version of the original $l$-completeness property defined in \cite{Willems1989}. To resolve the resulting inconsistencies, and also to address a wider system class, the procedure suggested in \cite{MoorRaisch1999} can be adapted in a straightforward way using the original $l$-completeness notion, capturing also time variant systems. This, not surprisingly, leads to realizations with time dependent next state relations. To address this, inspired by \cite{JuliusSchaft2005}, we have extended the well-known concepts of state property, memory span and $l$-completeness and have introduced asynchronous versions of these concepts. To clearly distinguish between the new, weaker versions and the original ones, the latter are referred to as synchronous properties.\\
Based on these extensions, we have proposed a new approximation technique,  called strongest asynchronous $l$-complete approximation. For systems with finite external signal space, it generates a finite state machine (FSM) as realization of the approximation. For time invariant systems, it produces the same approximation as \cite{MoorRaisch1999}, however, the mentioned inconsistencies are resolved.
The strongest asynchronous $l$-complete approximation of a given system is also a synchronous $l$-complete approximation, but not necessarily the strongest one. For strictly time invariant systems, we have shown that the concepts of strongest synchronous and strongest asynchronous $l$-complete approximations coincide.

\end{document}